\documentclass[AMA,Times2COL]{WileyNJDv5}
\usepackage{mathrsfs}
\newcommand{\preceqop}{\mathrel{\preccurlyeq}}

\usepackage{physics}

\usepackage{newunicodechar}
\newunicodechar{∞}{$\infty$}
\usepackage{tikz}
\usetikzlibrary{arrows.meta, positioning}
\makeatletter

\newcommand{\ipX}[3][]{%
  \mathinner{\left\langle #2, #3 \right\rangle_{\!#1}}%
}

\providecommand{\@history@dates}{}
\providecommand{\@DOI@text}{}
\providecommand{\@doi}{}

\providecommand{\oddfoot@titlepage@info}{}
\providecommand{\evenfoot@titlepage@info}{}
\providecommand{\oddhead@titlepage@info}{}
\providecommand{\evenhead@titlepage@info}{}

\makeatother
\articletype{Original Article}%

\received{}
\revised{}
\accepted{}

\begin{document}

\title{Operator-Theoretic Stability and Observer Synthesis for Parameter-Dependent Vlasov--Maxwell Dynamics}

\author[1]{Amadou Cissé}

\author[1,2]{Mohamed Boutayeb}


\authormark{A. Cissé and M. Boutayeb}


\address[1]{CRAN-UMR-CNRS-7039 University of Lorraine, F-57000, France.}

\address[2]{Université Internationale de Rabat - TicLab, Maroc.}

\corres{A. Cissé PhD, CRAN, Cosnes-et-Romain, 54400, France. \email{cissea87@yahoo.fr; amadou.cisse@univ-lorraine.fr}}



\abstract[Abstract]{
An operator--theoretic formulation is developed for the synthesis of parameter-dependent controllers and observers for the Vlasov--Maxwell system.
The linearized dynamics are modeled as a non-autonomous evolution system whose generators depend on measurable plasma quantities.
Well-posedness of the associated evolution family is established together with uniform growth bounds.
Parameter-dependent Lyapunov operators yield operator differential LMIs ensuring uniform exponential stability and observer convergence.
An $H_\infty$ extension provides disturbance attenuation conditions consistent with the intrinsic energy structure of the coupled Vlasov--Maxwell equations.
Galerkin projections lead to finite-dimensional LMIs consistent with the operator inequalities, enabling reliable numerical synthesis while
preserving the analytical structure of the original model.
Numerical results on a reduced Vlasov--Maxwell benchmark confirm the predicted convergence properties.}

\keywords{Vlasov--Maxwell system, parameter-dependent systems, operator differential
LMI, observer design, $H_\infty$ control, kinetic PDEs, non-autonomous evolution systems, hyperbolic systems.}


\maketitle

\renewcommand\thefootnote{}

\section{Introduction}
\label{sec:intro}

The Vlasov--Maxwell (VM) system provides a kinetic description of collisionless plasmas by coupling the evolution of the particle distribution function $f(t,x,v)$ on phase space $(x,v)\in\Omega\times\mathbb{R}^3$ with the electromagnetic fields $(E,B)$ governed by Maxwell’s equations.
This formulation captures the interaction between microscopic particle transport and electromagnetic wave propagation, which is central to plasma physics, astrophysics, and fusion research.

From a mathematical standpoint, the VM equations form a nonlinear hyperbolic kinetic--electromagnetic system combining phase-space transport, Maxwell propagation, nonlocal field interactions, and conservation laws.
While global weak solutions are known to exist under suitable assumptions, stabilization, state estimation, and feedback design remain challenging due to the infinite-dimensional nature of the dynamics and the strong coupling
between kinetic and electromagnetic components.

Control-oriented studies of kinetic models have progressed in recent years.
For the VM system, Glass and Han-Kwan~\cite{GlassHanKwan2012} established local exact controllability in a relativistic two-dimensional setting, while Weber~\cite{Weber2021} investigated optimal control
formulations together with associated sensitivity analysis.
More advanced observer and stabilization results are available for the Vlasov--Poisson subsystem, including distributed observers and feedback laws based on density or potential measurements~\cite{cisse2020observers,cisse2024software,cisse2024state}.
Related optimal control studies with external actuation were also reported in~\cite{knopf2020optimal}. More recently, a particle-in-cell Monte Carlo framework for controlling a Vlasov--Poisson plasma was developed
in~\cite{bartsch2024controlling}, extending earlier control strategies with an effective numerical implementation.
These works demonstrate the relevance of modern control techniques for kinetic equations, although often under reduced models or specific geometrical assumptions.

In parallel, parameter-dependent system formulations have become an effective tool for stability, robustness, and gain-scheduling analysis in time-varying systems~\cite{chesi2005polynomially,a2019computational,
jetto2010efficient,wu2001lpv}.
Operator-valued extensions have been developed for parabolic, hyperbolic, and transport-type partial differential equations~\cite{apkarian2000parameterized,coron2007strict,
mironchenko2017characterizations}.
These approaches make it possible to represent infinite-dimensional dynamics through operator families depending on measurable operating quantities while preserving the analytical structure of the underlying PDE model.
Despite these advances, no parameter-dependent operator formulation has been proposed for the VM equations, even though macroscopic plasma quantities naturally evolve in time and influence the linearized dynamics.

Such a parameter-dependent representation is physically  motivated by the temporal variability of macroscopic plasma indicators, including total kinetic energy, electromagnetic field energy, and current-density moments.
These quantities are known to fluctuate even near equilibrium and affect dispersion relations and stability thresholds in kinetic regimes~\cite{pezzi2018velocity,kamaletdinov2024nonlinear,
ghizzo2024collisionless}.
They enter the linearized VM operator through effective plasma frequency, temperature, or anisotropic pressure
terms~\cite{brizard2007foundations,arro2022spectral}, making them natural scheduling variables for a time-varying representation.

The contribution of this paper is not a direct reuse of existing parameter-dependent synthesis methods, but their extension to the coupled Vlasov--Maxwell setting. The resulting formulation combines phase-space transport, current-induced electromagnetic coupling, and parameter-varying field dynamics within a unified operator-theoretic framework. Particular attention is devoted to the functional setting associated with the unbounded velocity domain, the skew-adjoint structure of the Maxwell operator, and the consistency between operator-level inequalities and finite-dimensional Galerkin approximations.

Well-posedness of the non-autonomous evolution equation is established, together with uniform bounds on the associated evolution family.
Parameter-dependent Lyapunov operators yield operator differential LMIs ensuring uniform exponential stability, while a dual formulation provides parameter-dependent observers with guaranteed convergence.
Robustness with respect to disturbances is addressed through an $H_\infty$ extension, and the consistency of Galerkin approximations is proven, ensuring that finite-dimensional LMIs provide reliable numerical surrogates of the operator inequalities.
Numerical simulations illustrate the effectiveness of the proposed synthesis, validate the convergence properties predicted by the theory, and demonstrate applicability to discretized VM models of practical size.

The paper is organized as follows.
Section~\ref{sec:model} presents the controlled VM system and its parameter-dependent formulation.
Section~\ref{sec:wellposed} establishes well-posedness of the parameter-varying dynamics.
Sections~\ref{sec:lyapunov} and~\ref{sec:observer} develop the Lyapunov analysis and the construction of the observer operator. Robust $H_\infty$ performance and joint synthesis are addressed in Sections~\ref{sec:Hinf} and~\ref{sec:joint-Hinf}. Section~\ref{sec:numerical} reports numerical results, and Section~\ref{sec:conclusion} concludes the paper.

\section{Model and Parameter-Dependent Formulation}
\label{sec:model}

A collisionless plasma is described through the relativistic Vlasov--Maxwell equations on a spatial domain $\Omega$, taken in the analysis below as a periodic domain identified with the flat torus $\mathbb T^3$, with velocity variable $v\in\mathbb{R}^3$.
When boundary traces are recalled, $\Omega$ may equivalently be viewed as a smooth bounded domain with boundary $\partial\Omega$.
The state variables are the particle distribution function $f(t,x,v)$, the electric field $E(t,x)$, and the magnetic field $B(t,x)$.

\subsection{Controlled Vlasov--Maxwell equations}

The controlled dynamics take the form
\begin{equation}\label{eq:VM-control}
\left\{
\begin{aligned}
&\partial_t f + v\!\cdot\nabla_x f 
  + \big(E + v\times B\big)\!\cdot\nabla_v f = 0,\\[0.4em]
&\partial_t E - c^2\,\nabla_x\times B = - J + u_E,\\[0.4em]
&\partial_t B + \nabla_x\times E = 0,\\[0.4em]
&\nabla_x\!\cdot E = \varrho - \varrho_0,\qquad 
 \nabla_x\!\cdot B = 0,
\end{aligned}
\right.
\end{equation}
where the macroscopic quantities
\[
\varrho(t,x)= \int_{\mathbb{R}^3} f(t,x,v)\,dv,
\qquad
J(t,x)= \int_{\mathbb{R}^3} v\, f(t,x,v)\,dv
\]
denote the charge and current densities.

The input $u_E$ denotes an external actuation channel entering the electromagnetic subsystem through the electric-field equation. It represents an
equivalent source term generated by boundary devices, coils, antennas, or distributed current actuation. This modeling choice preserves the intrinsic
Maxwell coupling, while the magnetic field evolves through the standard electromagnetic dynamics.

\medskip
Actuation in Vlasov--Maxwell control models is commonly provided through external electromagnetic devices such as radio-frequency antennas, magnetic coils, or electrostatic actuators.
These devices modify the Maxwell subsystem through controlled currents or boundary fields, and their effect on the kinetic dynamics propagates through the Lorentz force; see~\cite{svidzinski2024full,zaar2025enhanced}.

\medskip
The initial configuration is prescribed through
\[
f(0,x,v)=f_0(x,v),\qquad
E(0,x)=E_0(x),\qquad
B(0,x)=B_0(x),
\]
together with boundary conditions ensuring compatibility with Maxwell constraints.

Throughout the paper, periodic boundary conditions are assumed on the spatial domain $\Omega$, which is identified with a torus. This setting avoids boundary
compatibility issues and yields standard semigroup generation properties for the transport and Maxwell operators.

The present analysis is therefore restricted to the periodic framework and does not address reflecting or absorbing boundary conditions. In kinetic models with physical boundaries, additional phenomena such as boundary-induced singularity formation or concentration effects may arise, even for linearized Vlasov--Maxwell dynamics. The periodic setting adopted here excludes such boundary mechanisms and provides a mathematically consistent framework for the operator-theoretic analysis developed below.

\subsection{Functional Setting}
Because the velocity domain is unbounded, the kinetic component is not taken in a plain unweighted space when velocity moments are involved. We introduce the weighted kinetic space
\[
L^2_m(\Omega\times\mathbb{R}^3)
:=
\left\{
f:\int_{\Omega}\int_{\mathbb{R}^3}
|f(x,v)|^2m(v)\,dv\,dx<\infty
\right\},
\]
where
\[
m(v)=(1+|v|^2)^s,
\qquad s>\frac52.
\]
The exponent $s>5/2$ ensures that both the charge density and the current density associated with $f$ are well defined in $L^2(\Omega)$ and $L^2(\Omega;\mathbb R^3)$, respectively. Indeed, by Cauchy's inequality,
\[
\left\|\int_{\mathbb R^3} f(\cdot,v)\,dv\right\|_{L^2_x}
\le
\left(\int_{\mathbb R^3}\frac{1}{m(v)}\,dv\right)^{1/2}
\|f\|_{L^2_m},
\]
and
\[
\left\|\int_{\mathbb R^3} v f(\cdot,v)\,dv\right\|_{L^2_x}
\le
\left(\int_{\mathbb R^3}\frac{|v|^2}{m(v)}\,dv\right)^{1/2}
\|f\|_{L^2_m}.
\]

Let
\[
X_f = L^2_m(\Omega\times\mathbb{R}^3),\qquad 
X_E = L^2(\Omega;\mathbb{R}^3),\qquad 
X_B = L^2(\Omega;\mathbb{R}^3),
\]
and define the product state space
\[
\mathscr{X} := X_f \times X_E \times X_B.
\]
It is equipped with the natural Vlasov--Maxwell energy inner product
\[
\langle x_1 , x_2 \rangle_{\mathscr{X}}
= 
\langle f_1 , f_2 \rangle_{L^2_m}
+
\langle E_1 , E_2 \rangle_{L^2_x}
+
c^2\langle B_1 , B_2 \rangle_{L^2_x},
\]
and norm
\[
\|x\|_{\mathscr{X}}^2
=
\|f\|_{L^2_m}^2
+
\|E\|_{L^2_x}^2
+
c^2\|B\|_{L^2_x}^2 .
\]
This norm is equivalent to the standard product $L^2$ norm and is the one used below when adjoints and skew-adjointness are considered.

\paragraph*{Transport domain}
The kinetic transport operator is defined on the graph domain
\begin{equation}\label{eq:dom}
D(T)
=
\left\{
f\in X_f:\ v\!\cdot\nabla_x f\in X_f
\right\}.
\end{equation}
The corresponding transport--Maxwell operator used in the well-posedness analysis is introduced in Section~\ref{sec:wellposed}.
\paragraph*{The space $H(\mathrm{curl};\Omega)$ and tangential traces}
Consider
\[
H(\mathrm{curl};\Omega)=
\big\{\,u\in L^2(\Omega;\mathbb{R}^3):\
\nabla_x\times u\in L^2(\Omega;\mathbb{R}^3)\,\big\},
\]
with norm
\[
\|u\|_{H(\mathrm{curl})}^2
=
\|u\|_{L^2}^2+\|\nabla\times u\|_{L^2}^2 .
\]
For bounded Lipschitz domains, the tangential trace
$\gamma_t(u)=n\times u|_{\partial\Omega}$ defines a continuous mapping into
$H^{-1/2}(\mathrm{div}_{\Gamma},\partial\Omega)$.
Green's formula in this setting reads
\begin{equation}\label{eq:greens-curl}
\langle \nabla\times u , v \rangle_{L^2}
-
\langle u , \nabla\times v \rangle_{L^2}
=
\int_{\partial\Omega} (n\times u)\cdot v\, \mathrm{d}S ,
\quad
u,v\in H(\mathrm{curl};\Omega).
\end{equation}

\paragraph*{Maxwell operator}
Define the associated electromagnetic operator
\begin{equation}\label{eq:A_M}
\mathcal{A}_M
\begin{bmatrix} E \\[0.2em] B \end{bmatrix}
=
\begin{bmatrix}
0 & c^{2}\,\nabla_x \!\times \\[0.3em]
-\,\nabla_x \!\times & 0
\end{bmatrix}
\begin{bmatrix} E \\[0.2em] B \end{bmatrix}
=
\begin{bmatrix}
c^{2}\,\nabla_x \!\times B \\[0.3em]
-\,\nabla_x \!\times E
\end{bmatrix},
\end{equation}
$(E,B)\in \mathcal{D}(\mathcal{A}_M)$, with \(
\mathcal D(\mathcal A_M)
=
H(\mathrm{curl};\Omega)
\times
H(\mathrm{curl};\Omega).
\) 
The operator $\mathcal A(\rho)$ is defined on the common dense domain
\[
\mathcal D(\mathcal A)=
D(T)\times H(\mathrm{curl};\Omega)\times H(\mathrm{curl};\Omega),
\]
supplemented with the periodic boundary conditions specified above.

\begin{definition}
Given a densely defined linear operator
$T:\mathcal{D}(T)\subset \mathscr{H}\to \mathscr{H}$ on a Hilbert space
$(\mathscr{H},\langle\cdot,\cdot\rangle)$, its adjoint $T^{\star}$ is defined by
\[
\mathcal{D}(T^{\star})
:=
\Big\{y\in \mathscr{H}:\exists\,z\in\mathscr{H}\ \text{s.t.}\ 
\langle Tx,y\rangle=\langle x,z\rangle\ \ \forall x\in\mathcal{D}(T)\Big\},
\]
$T^{\star}y:=z.$
\end{definition}
In what follows, $T^{\star}$ is always understood with respect to the chosen energy inner product.

The following result identifies the skew-adjoint structure of the Maxwell operator.

\begin{proposition}
\label{prop:AM-skew}
Let $X_E$, $X_B$, and equip $X_E\times X_B$ with the electromagnetic energy inner product
\[
\langle (E,B),(E',B')\rangle_{\mathcal E}
=
\langle E,E'\rangle_{L^2}
+
c^2\langle B,B'\rangle_{L^2}.
\]
Assume boundary conditions on $(E,B)$ that make the boundary terms in Green's identity for $\nabla\times$ vanish, such as periodic or perfectly conducting boundaries. Then
\[
\mathcal{A}_M^{\star}=-\,\mathcal{A}_M .
\]
Consequently, $\mathcal{A}_M$ generates a unitary $C_0$-group on $X_E\times X_B$.
\end{proposition}

\begin{proof}
Take $(E,B),(E',B')\in\mathcal{D}(\mathcal{A}_M)$.  
Using~\eqref{eq:A_M} and the electromagnetic energy inner product gives
\begin{align*}
\big\langle \mathcal{A}_M(E,B),(E',B')\big\rangle_{\mathcal E}
&=
\langle c^{2}\nabla_x\times B, E'\rangle_{L^{2}}
+
c^2\langle -\nabla_x\times E, B'\rangle_{L^{2}}\\
&=
c^{2}\,\langle \nabla_x\times B, E'\rangle_{L^{2}}
-
c^2\langle \nabla_x\times E, B'\rangle_{L^{2}} .
\end{align*}
Green's identity yields
\[
\langle \nabla_x\times B , E' \rangle_{L^{2}}
=
\langle B , \nabla_x\times E' \rangle_{L^{2}}
+
\int_{\partial\Omega} (n\times B)\!\cdot E'\, dS,
\]
and
\[
\langle \nabla_x\times E , B' \rangle_{L^{2}}
=
\langle E , \nabla_x\times B' \rangle_{L^{2}}
+
\int_{\partial\Omega} (n\times E)\!\cdot B'\, dS .
\]
Hence
\begin{align*}
\big\langle \mathcal{A}_M(E,B),(E',B')\big\rangle_{\mathcal E}
&=
c^{2}\,\langle B,\nabla_x\times E' \rangle_{L^{2}}
-
c^2\langle E,\nabla_x\times B' \rangle_{L^{2}}\\ 
&+
c^{2}\!\int_{\partial\Omega} (n\times B)\!\cdot E'\,dS\\
&-
c^2\!\int_{\partial\Omega} (n\times E)\!\cdot B'\,dS .
\end{align*}
For perfectly conducting boundaries,
\(
n\times E =0,\quad n\cdot B=0
\quad\text{on }\partial\Omega,
\)
and the boundary terms vanish. In the periodic setting, opposite faces cancel as well. Thus
\[
\big\langle \mathcal{A}_M(E,B),(E',B')\big\rangle_{\mathcal E}
=
c^{2}\,\langle B,\nabla_x\times E' \rangle_{L^{2}}
-
c^2\langle E,\nabla_x\times B' \rangle_{L^{2}} .
\]
On the other hand,
\begin{align*}
\big\langle (E,B), \mathcal{A}_M(E',B') \big\rangle_{\mathcal E}
&=
\langle E,c^2\nabla_x\times B' \rangle_{L^{2}}
+
c^2\langle B,-\nabla_x\times E' \rangle_{L^{2}}\\
&=
c^2\langle E,\nabla_x\times B' \rangle_{L^{2}}
-
c^2\langle B,\nabla_x\times E' \rangle_{L^{2}} .
\end{align*}
Therefore,
\[
\big\langle \mathcal{A}_M(E,B),(E',B')\big\rangle_{\mathcal E}
=
-\big\langle (E,B),\mathcal{A}_M(E',B')\big\rangle_{\mathcal E},
\]
so $\mathcal{A}_M\subset -\mathcal{A}_M^\star$.
Classical Maxwell theory on PEC or periodic domains gives the maximality condition, hence $\mathcal{A}_M^\star=-\mathcal{A}_M$.
The unitary group property follows from Stone's theorem.
\end{proof}

\begin{remark}
The electromagnetic energy inner product may equivalently be written as
\[
\langle (E,B),(E',B')\rangle_{\mathcal{E}}
=
\varepsilon_0\,\langle E,E'\rangle_{L^2}
+
\mu_0^{-1}\,\langle B,B'\rangle_{L^2},
\quad
c^{2}=(\varepsilon_0\mu_0)^{-1}.
\]
This inner product is equivalent to the standard $L^2$ one, with a diagonal positive scaling on the $(E,B)$ components.
Since this scaling is bounded, self-adjoint, and coercive, the skew-adjoint structure is preserved in the corresponding energy space.
\end{remark}

\subsection{Nominal trajectory and perturbation variables}

Let $(\bar f,\bar E,\bar B)$ denote a sufficiently regular nominal solution of \eqref{eq:VM-control}, associated with prescribed controls
$\bar u_E$.
Define the perturbation variables
\[
\delta f = f - \bar f,\qquad
\delta E = E - \bar E,\qquad
\delta B = B - \bar B,
\]
and introduce the charge and current variations
\[
\delta\varrho(t,x)
=
\int_{\mathbb{R}^3}\delta f(t,x,v)\,dv,
\qquad
\delta J(t,x)
=
\int_{\mathbb{R}^3} v\,\delta f(t,x,v)\,dv.
\]
The control is decomposed as
\[
u_E = \bar u_E + \delta u_E.
\]
Linearization of~\eqref{eq:VM-control} around $(\bar f,\bar E,\bar B)$ yields
\begin{equation}\label{eq:VM-linearized}
\left\{
\begin{aligned}
&\partial_t \delta f
+ v\cdot\nabla_x \delta f
+ (\bar E + v\times \bar B)\cdot\nabla_v \delta f
+ (\delta E + v\times \delta B)\cdot\nabla_v \bar f = 0,\\[0.3em]
&\partial_t \delta E - c^2\nabla_x\times \delta B
= -\,\delta J + \delta u_E,\\[0.3em]
&\partial_t \delta B + \nabla_x\times \delta E
= 0,\\[0.3em]
&\nabla_x\!\cdot \delta E = \delta\varrho,
\qquad
\nabla_x\!\cdot \delta B = 0.
\end{aligned}
\right.
\end{equation}

System~\eqref{eq:VM-linearized} couples phase-space transport with the linearized Maxwell equations.
The terms
$(\bar E + v\times\bar B)\cdot\nabla_v \delta f$ and
$(\delta E + v\times\delta B)\cdot\nabla_v \bar f$
represent, respectively, transport along the Lorentz force induced by the
nominal fields and sensitivity of the particle distribution to electromagnetic
perturbations.
The constraints $\nabla_x\!\cdot\delta E = \delta\varrho$ and $\nabla_x\!\cdot\delta B = 0$ enforce Gauss' laws at the linearized level.

\subsection{Parameter scheduling and parameter-dependent operator structure}


The coefficients in~\eqref{eq:VM-linearized} depend on a nominal operating state parameterized by a physically measurable parameter.
Assume that the parameter trajectory
$\rho:[0,\infty)\to\mathbb P$ is piecewise $\mathcal C^1$, where $\mathbb P\subset\mathbb R^p$ is a compact admissible parameter set with nonempty interior. Typical components of $\rho(t)$ include averaged kinetic energy, plasma frequency, or electromagnetic field magnitudes.

The nominal fields
\(
(\bar f(\rho),\bar E(\rho),\bar B(\rho))
\)
are interpreted as a family of frozen operating regimes parameterized by \(\rho\in\mathbb P\), rather than as arbitrary explicitly time-dependent background trajectories. Consequently, for each fixed parameter value
\(\rho\), the associated operator \(\mathcal A(\rho)\) is autonomous.
The non-autonomous character of the evolution is induced solely through the time variation of the parameter trajectory \(\rho(\cdot)\), which continuously connects the corresponding operating regimes.

\smallskip
Define the perturbation state and control variables as
\[
X(t)=
\begin{bmatrix}
\delta f(t)\\[0.2em]
\delta E(t)\\[0.2em]
\delta B(t)
\end{bmatrix},
\qquad
u(t)=
\begin{bmatrix}
\delta u_E(t)\\[0.2em]
0
\end{bmatrix}.
\]
Then the variational system~\eqref{eq:VM-linearized} admits the abstract parameter-dependent representation on the state space $\mathscr{X}$
\begin{equation}\label{eq:LPV-abstract}
\dot X(t)=\mathcal{A}(\rho(t))X(t)+\mathcal{B}(\rho(t))u(t),
\qquad
y(t)=\mathcal{C}(\rho(t))X(t).
\end{equation}

The linear operators $\mathcal{A}(\rho)$, $\mathcal{B}(\rho)$, and $\mathcal{C}(\rho)$ encode the transport dynamics, Lorentz coupling, and Maxwell equations.
Their block representations are
\begin{align*}
\mathcal{A}(\rho)=
\begin{bmatrix}
-\,v\!\cdot\nabla_x
-(\bar E(\rho)+v\times\bar B(\rho))\!\cdot\nabla_v
&
-\nabla_v \bar f(\rho)\!\cdot
&
0\\[0.4em]
-\,J[\cdot]
&
0
&
c^2\nabla_x\times\\[0.4em]
0
&
-\nabla_x\times
&
0
\end{bmatrix},
\end{align*}
\(\mathcal{B}(\rho)=
\begin{bmatrix}
0\\[0.2em]
I\\[0.2em]
0
\end{bmatrix}.\)

The current functional is
\[
J[\delta f](x)=\int_{\mathbb{R}^3} v\,\delta f(x,v)\,dv.
\]
By the choice $X_f=L^2_m(\Omega\times\mathbb R^3)$ with $s>5/2$, this functional is bounded from $X_f$ into $L^2(\Omega;\mathbb R^3)$.

\smallskip
The observation operator $\mathcal{C}(\rho)$ represents the measurement process and maps the state perturbation
$X=(\delta f,\delta E,\delta B)$ to observable macroscopic quantities.
Depending on the sensing configuration, the output space is taken as $Y=L^2(\Omega;\mathbb{R}^3)$, and
$\mathcal{C}(\rho)\in\mathcal{L}(\mathscr{X},Y)$ denotes a bounded linear projection or restriction operator extracting, for instance, components of the electric field or charge-related quantities.
The control space is chosen as
\[
U=L^2(\Omega;\mathbb{R}^3)
\]
for electromagnetic actuation.

\smallskip
The operator $\mathcal{A}(\rho)$ acts on the common dense domain $\mathcal{D}(\mathcal{A})$
independent of $\rho$.
The Gauss constraints
\[
\nabla_x\!\cdot\delta E=\int_{\mathbb{R}^3}\delta f\,dv,
\qquad
\nabla_x\!\cdot\delta B=0
\]
are imposed as compatibility conditions on admissible initial data.

\smallskip
Smoothness of $(\bar f,\bar E,\bar B)$ ensures that the mapping
\[
\rho\mapsto \mathcal{A}(\rho)
\]
is continuous from $\mathbb{P}$ into
$\mathcal{L}(\mathcal{D}(\mathcal{A}),\mathscr{X})$, and that each operator $\mathcal{A}(\rho)$ generates a strongly continuous semigroup on $\mathscr{X}$.
Moreover, the corresponding flow preserves the Gauss constraints, so that the associated subspace $\mathscr{X}_{\mathrm G}\subset\mathscr{X}$ is invariant
under the dynamics.

\smallskip
The next section establishes well-posedness of the non-autonomous system
\eqref{eq:LPV-abstract}.

\section{Well-posedness of the Parameter-Dependent Vlasov--Maxwell System}
\label{sec:wellposed}

The operator family $\mathcal{A}(\rho)$ introduced in~\eqref{eq:LPV-abstract} depends on a measurable scheduling parameter and describes the linearized
kinetic--electromagnetic dynamics along a nominal trajectory.  
This section establishes existence, uniqueness, and regularity of the corresponding non-autonomous evolution.

\subsection{Assumptions on the operator family}

The following hypotheses specify the structural and regularity properties required for the parameter-dependent operator family $\{\mathcal{A}(\rho)\}_{\rho\in\mathbb{P}}$ to generate a well-posed non-autonomous evolution system.

\begin{assumption}[Structural properties]\label{ass:A}
For every $\rho\in\mathbb{P}$, the operator $\mathcal{A}(\rho):\mathcal{D}(\mathcal{A})\subset\mathscr{X}\to\mathscr{X}$
satisfies:
\begin{enumerate}
    \item[\rm(i)] $\mathcal{D}(\mathcal{A})$ is dense in $\mathscr{X}$ and does not depend on $\rho$;
    \item[\rm(ii)] each frozen operator $\mathcal{A}(\rho)$ generates a $\mathcal{C}_0$-semigroup on $\mathscr{X}$;
    \item[\rm(iii)] the family $\{\mathcal{A}(\rho)\}_{\rho\in\mathbb P}$ satisfies the Kato stability conditions uniformly on $\mathbb P$;
    \item[\rm(iv)] the mapping
    $\rho\mapsto\mathcal{A}(\rho)$ is continuous from $\mathbb P$ into
    $\mathcal{L}(\mathcal{D}(\mathcal{A}),\mathscr{X})$.
\end{enumerate}
\end{assumption}

\begin{assumption}[Parameter trajectories]\label{ass:rho}
The scheduling map $\rho:[0,\infty)\to\mathbb{P}$ is piecewise $C^{1}$ and
satisfies 
\[
\dot\rho(t)\in\mathcal{V},
\qquad 
\|\dot\rho(t)\|\le \nu_{\max}
\quad\text{for a.e. } t\ge 0 ,
\]
where $\mathcal{V}$ is a compact set of admissible parameter rates, and $\nu_{\max}>0$ denotes its maximal radius, i.e.,
\[
\nu_{\max}:=\max_{\nu\in\mathcal{V}}\|\nu\|.
\]
\end{assumption}

\begin{assumption}[Regularity of bounded families]
\label{ass:bounded-families}
For every $\rho\in\mathbb{P}$, the observation operator
$\mathcal{C}(\rho)$ and the Lyapunov-operator derivatives
$\partial_{\rho_i}\mathcal{P}_o(\rho)$, $i=1,\dots,p$, are bounded on $\mathscr{X}$ and depend continuously on $\rho$ in the operator norm.
Moreover,
\[
\sup_{\rho\in\mathbb{P}}
\Big(
\|\mathcal{C}(\rho)\|_{\mathcal{L}(\mathscr{X},Y)}
+
\sum_{i=1}^{p}
\|\partial_{\rho_i}\mathcal{P}_o(\rho)\|_{\mathcal{L}(\mathscr{X})}
\Big)<\infty .
\]
\end{assumption}

\begin{assumption}[Initial conditions]\label{ass:init}
The control input $u(\cdot)$ lies in $L^2_{\rm loc}([0,\infty);U)$ and the initial condition satisfies $X(0)=X_0\in\mathscr{X}$.
\end{assumption}

\subsection{Semigroup generation}

The unforced kinetic--electromagnetic dynamics are represented by the operator family
\begin{equation}\label{eq:opA_0}
\mathcal{A}_0(\rho)=
\begin{bmatrix}
-\,v\!\cdot\!\nabla_x
-
(\bar E(\rho)+v\times\bar B(\rho))\!\cdot\!\nabla_v

& 0 & 0\\[0.2em]
0 & 0 & c^2\nabla_x\times\\[0.2em]
0 & -\nabla_x\times & 0
\end{bmatrix},
\end{equation}
acting on the Hilbert space
\(
\mathscr X
\)
equipped with the energy inner product introduced in
Section~\ref{sec:model}.

For each fixed parameter \(\rho\in\mathbb P\), the nominal fields
\(
(\bar f(\rho),\bar E(\rho),\bar B(\rho))
\)
are assumed independent of time. Hence the non-autonomous character of the evolution originates exclusively from the parameter trajectory \(\rho(\cdot)\), and not from an explicitly time-dependent background state.
This setting is compatible with the evolution-family framework adopted in Theorem~\ref{thm:wellposed}.

The domain is defined by
\[
\mathcal D(\mathcal A_0)
=
D(T)\times H(\mathrm{curl};\Omega)\times H(\mathrm{curl};\Omega),
\]
where the transport domain \(D(T)\) is defined in~\eqref{eq:dom}.

The transport operator therefore contains the full first-order kinetic transport structure
\[
T_\rho f
=
-\,v\!\cdot\!\nabla_x f
-
(\bar E(\rho)+v\times\bar B(\rho))\!\cdot\!\nabla_v f .
\]
This point is essential since the operator
\(
(\bar E(\rho)+v\times\bar B(\rho))\cdot\nabla_v
\)
is itself unbounded on \(X_f\), and therefore cannot be treated as a bounded perturbation of
\(
-v\cdot\nabla_x
\).

In the periodic setting considered throughout this work, no boundary condition other than spatial periodicity is required.

Notice that the condition
\(
f\in H^1_{x,v}(\Omega\times\mathbb R^3)
\)
does not in general imply
\(
v\!\cdot\!\nabla_x f\in L^2_{x,v},
\)
due to the unboundedness of the velocity variable. For this reason, the transport operator is defined on its natural graph domain \(D(T)\).

Assume moreover that
\(
\bar E(\rho),\bar B(\rho)
\in
W^{1,\infty}(\Omega;\mathbb R^3),
\;
\forall \rho\in\mathbb P,
\)
and that the mapping
\(
\rho\mapsto (\bar E(\rho),\bar B(\rho))
\)
is continuous on \(\mathbb P\).
Under these assumptions, \(T_\rho\) defines a closed transport operator on \(X_f\) with common dense domain \(D(T)\).

The skew-adjointness of the Maxwell block is a key ingredient for the semigroup generation result.

\begin{proposition}
\label{prop:A0-skew}
For each fixed parameter \(\rho\in\mathbb P\), the operator
\(\mathcal A_0(\rho)\) defined in~\eqref{eq:opA_0} generates a strongly continuous semigroup on \(\mathscr X\). Moreover, the Maxwell block is skew-adjoint and generates a unitary \(C_0\)-group on the electromagnetic
energy space.
\end{proposition}

\begin{proof}
Fix \(\rho\in\mathbb P\).
The kinetic operator
\(
T_\rho
=
-\,v\!\cdot\!\nabla_x
-
(\bar E(\rho)+v\times\bar B(\rho))\!\cdot\!\nabla_v
\)
is a first-order transport operator with bounded coefficients on the periodic phase space.

Furthermore,
\(
\nabla_v\!\cdot(\bar E(\rho)+v\times\bar B(\rho))=0,
\)
since \(\bar E(\rho)\) is independent of \(v\) and
\(
\nabla_v\!\cdot(v\times\bar B(\rho))=0.
\)
Hence the transport field is divergence free in the velocity variable, which preserves the conservative structure of the Vlasov dynamics.

Standard transport theory on periodic domains therefore yields generation of a strongly continuous semigroup on
\(
X_f
\).
The Maxwell operator has already been shown to be skew-adjoint in Proposition~\ref{prop:AM-skew}. Hence the block
\[
\begin{bmatrix}
0 & c^{2}\nabla_x\times\\
-\nabla_x\times & 0
\end{bmatrix}
\]
is skew-adjoint on
\(
H(\mathrm{curl};\Omega)\times H(\mathrm{curl};\Omega)
\),
with respect to the electromagnetic energy inner product.

The kinetic and electromagnetic blocks act on orthogonal components of \(\mathscr X\), and their domains form a direct product. Consequently, \(\mathcal A_0(\rho)\) generates a strongly continuous semigroup on \(\mathscr X\).
\end{proof}

The remaining coupling terms of the linearized dynamics define lower-order perturbations. More precisely, define
\[
\mathcal G(\rho)
:=
\mathcal A(\rho)-\mathcal A_0(\rho).
\]
Then the only remaining kinetic coupling contribution is
\(
(\delta E+v\times\delta B)\cdot\nabla_v\bar f(\rho),
\)
which acts multiplicatively on the perturbation variables.
Assuming
\(
\nabla_v\bar f(\rho)
\in
L^\infty(\Omega\times\mathbb R^3),
\)
with sufficient decay in the velocity variable, this term defines a bounded operator on \(\mathscr X\).

Moreover, the weighted kinetic space
\(X_f\), with \(s>\frac52\), ensures that the current functional
\[
J[\delta f](x)=\int_{\mathbb R^3}v\,\delta f(x,v)\,dv
\]
is well defined as an element of \(L^2(\Omega;\mathbb R^3)\). Consequently,
\(\mathcal G(\rho)\) acts as a bounded perturbation on the common domain
\(
\mathcal D(\mathcal A)
=
\mathcal D(\mathcal A_0).
\)
Therefore, for each fixed \(\rho\in\mathbb P\), the operator \(\mathcal A(\rho)\) generates a $\mathcal C_0$-semigroup \((T_\rho(t))_{t\ge0}\).

Furthermore, by continuity of the coefficients with respect to \(\rho\) and compactness of \(\mathbb P\), there exist constants \(M\ge1\) and \(\omega\in\mathbb R\), independent of \(\rho\), such that
\[
\|T_\rho(t)\|_{\mathcal L(\mathscr X)}
\le
M e^{\omega t},
\qquad t\ge0.
\]

\subsection{Non-autonomous evolution for time-varying parameters}

The parameter-dependent system
\begin{equation}
\dot X(t)=\mathcal{A}(\rho(t))X(t)+\mathcal{B}(\rho(t))u(t),
\qquad X(0)=X_0,
\end{equation}
admits a unique mild solution constructed through a non-autonomous evolution family.

The following theorem establishes well-posedness of the parameter-dependent evolution system.

\begin{theorem}
\label{thm:wellposed}
Under Assumptions~\ref{ass:A}--\ref{ass:init}, for every admissible parameter trajectory $\rho(\cdot)$ and every input $u\in L^2_{\mathrm{loc}}([0,\infty);U)$, there exists a unique mild solution
\(
X(\cdot)\in C([0,\infty);\mathscr{X})
\)
of the parameter-dependent system, given for all $t\ge s\ge0$ by
\begin{equation}\label{eq:evolution-family}
X(t)=\mathcal{S}_{\rho}(t,s)X(s)
+\int_s^t
\mathcal{S}_\rho(t,\tau)\mathcal{B}(\rho(\tau))u(\tau)\,d\tau,
\end{equation}
where $(\mathcal{S}_{\rho})_{t\ge s}$ is a strongly continuous evolution family
satisfying
\[
\|\mathcal{S}_{\rho}(t,s)\|_{\mathcal{L}(\mathscr{X})}
\le M e^{\omega(t-s)}, \qquad t\ge s\ge0.
\]
\end{theorem}

\begin{proof}
By Assumption~\ref{ass:A}, the family
$\{\mathcal{A}(\rho)\}_{\rho\in\mathbb P}$ has a common dense domain, each frozen operator generates a $\mathcal C_0$-semigroup, and the Kato stability bounds are uniform on $\mathbb P$.
Moreover, the continuity of
\[
\rho\mapsto\mathcal A(\rho)
\]
from $\mathbb P$ into
$\mathcal L(\mathcal D(\mathcal A),\mathscr X)$, together with the piecewise
$C^1$ regularity of $\rho(\cdot)$, implies that
$t\mapsto\mathcal A(\rho(t))x$ is continuous for every
$x\in\mathcal D(\mathcal A)$.

Kato's non-autonomous evolution theorem~\ref{thm:Kato} then yields a unique strongly continuous evolution family
$(\mathcal S_{\rho}(t,s))_{t\ge s\ge0}$ associated with
$\mathcal A(\rho(t))$, satisfying the composition rule, identity property, and uniform exponential bound.

For any admissible input $u$, the integral term in
\eqref{eq:evolution-family} is well defined in $\mathscr{X}$, and the variation-of-constants formula provides a mild solution. Uniqueness follows
from linearity and the exponential bound on $\mathcal{S}_{\rho}$, while continuity of $X(\cdot)$ follows from the strong continuity of the evolution
family.
\end{proof}

\begin{remark}
If $\rho(\cdot)$ is $C^1$ and
$\rho\mapsto\mathcal A(\rho)$ is $C^1$ from $\mathbb P$ into $\mathcal L(\mathcal D(\mathcal A),\mathscr X)$, then solutions with initial data in $\mathcal D(\mathcal A)$ are strong solutions on intervals where $\rho$ is $C^1$.
\end{remark}

\begin{remark}
The constants $(M,\omega)$ bounding the evolution family are uniform in $\rho$ by the Kato stability assumption and compactness of $\mathbb P$.
This uniformity is crucial for the Lyapunov and $H_\infty$ estimates developed in later sections.
\end{remark}

The parameter-dependent variational model of the Vlasov--Maxwell system generates a uniformly bounded non-autonomous evolution for all admissible parameter trajectories.
This establishes an analytical foundation for the stability, observer, and $H_\infty$ synthesis results derived in the following sections.

\section{Parameter-Dependent Lyapunov Functionals and Operator Differential LMIs}
\label{sec:lyapunov}

Let $\mathcal{P}:\mathbb{P}\to\mathcal{L}(\mathscr{X})$ denote a family of bounded, self-adjoint, coercive operators.  
There exist constants $0<\underline{p}\le \overline{p}$ such that
\begin{equation}\label{eq:P-bounds}
\underline{p}\,\|X\|_{\mathscr{X}}^{2}
\le
\langle X,\mathcal{P}(\rho)X\rangle
\le
\overline{p}\,\|X\|_{\mathscr{X}}^{2},
\qquad
\forall\,X\in\mathscr{X},\;\rho\in\mathbb{P}.
\end{equation}
The parameter-dependent Lyapunov functional is
\begin{equation}\label{eq:Lyap-func}
V_\rho(X) := \langle X,\mathcal{P}(\rho)X\rangle .
\end{equation}
The mapping $\rho\mapsto\mathcal{P}(\rho)$ is assumed strongly differentiable, with bounded derivatives
$\partial_{\rho_i}\mathcal{P}(\rho)\in\mathcal{L}(\mathscr{X})$.

\subsection{Derivative along parameter-dependent trajectories}

Consider the homogeneous parameter-dependent system
\begin{equation}\label{eq:LPV-homog}
\dot X(t)=\mathcal{A}(\rho(t))X(t),
\qquad X(0)=X_0,
\end{equation}
where $\rho(\cdot)$ satisfies Assumption~\ref{ass:rho}.  
For each $\rho\in\mathbb{P}$, the operator $\mathcal{P}(\rho)$ is bounded, self-adjoint, coercive, and continuously differentiable with respect to~$\rho$.

Evaluating $V_\rho$ along a trajectory $t\mapsto X(t)$ yields
\[
V(t)=\langle X(t),\mathcal{P}(\rho(t))X(t)\rangle .
\]
Applying the chain rule gives
\[
\dot V(t)
=
\langle \dot X(t),\mathcal{P}(\rho(t))X(t)\rangle
+
\langle X(t),\mathcal{P}(\rho(t))\dot X(t)\rangle
+
\langle X(t),\dot{\mathcal{P}}(\rho(t))X(t)\rangle .
\]
Substituting $\dot X(t)=\mathcal{A}(\rho(t))X(t)$ gives
\[
\langle \dot X(t),\mathcal{P}(\rho(t))X(t)\rangle
=
\langle X(t),\mathcal{A}(\rho(t))^{\star}
\mathcal{P}(\rho(t))X(t)\rangle,
\]
and
\[
\langle X(t),\mathcal{P}(\rho(t))\dot X(t)\rangle
=
\langle X(t),\mathcal{P}(\rho(t))
\mathcal{A}(\rho(t))X(t)\rangle .
\]
Since $\rho\mapsto\mathcal{P}(\rho)$ is $C^1$ in the strong operator topology,
\[
\dot{\mathcal{P}}(\rho(t))
=
\sum_{i=1}^{p}\dot\rho_i(t)\,
\partial_{\rho_i}\mathcal{P}(\rho(t)),
\]
and therefore
\[
\langle X(t),\dot{\mathcal{P}}(\rho(t))X(t)\rangle
=
\sum_{i=1}^{p}\dot\rho_i(t)\,
\langle X(t),\partial_{\rho_i}\mathcal{P}(\rho(t))X(t)\rangle .
\]
Collecting all terms and omitting explicit time dependence gives
\begin{equation}\label{eq:Lyap-derivative}
\dot V_\rho(X)
=
\big\langle X,\,
\big(
\mathcal{A}(\rho)^{\star}\mathcal{P}(\rho)
+
\mathcal{P}(\rho)\mathcal{A}(\rho)
\big)X\big\rangle
+
\sum_{i=1}^{p}\dot\rho_i\,
\big\langle X,\partial_{\rho_i}\mathcal{P}(\rho)X\big\rangle .
\end{equation}
The first contribution reflects the frozen-parameter dynamics, whereas the second accounts for variations of the scheduling parameter.

\subsection{Operator differential LMI}

Uniform exponential stability is characterized by an operator inequality that extends the classical Lyapunov condition.

\begin{definition}[Operator differential LMI]\label{def:ODLMI}
A pair $(\mathcal{P},\mathcal{K})$ is said to satisfy the operator differential LMI if there exists $\alpha>0$ such that, for every $\rho\in\mathbb{P}$ and every admissible $\nu\in\mathcal{V}$,
\begin{equation}\label{eq:ODLMI}
\mathcal{A}_K(\rho)^{\star}\mathcal{P}(\rho)
+
\mathcal{P}(\rho)\mathcal{A}_K(\rho)
+
\sum_{i=1}^{p}\nu_i\,\partial_{\rho_i}\mathcal{P}(\rho)
\preceq
-2\alpha\, \mathcal{P}(\rho),
\end{equation}
where
\(
\mathcal{A}_K(\rho)
=
\mathcal{A}(\rho)+\mathcal{B}(\rho)\mathcal{K}(\rho),
\)
and $\preceq$ denotes the operator inequality on $\mathscr{X}$.
\end{definition}

\begin{remark}
Inequality~\eqref{eq:ODLMI} incorporates the parameter variations through the term $\sum_i\nu_i\,\partial_{\rho_i}\mathcal{P}(\rho)$.
If $\mathcal{P}(\rho)$ is constant, one recovers the standard operator Lyapunov inequality
\[
\mathcal{A}_K^{\star}\mathcal{P}
+
\mathcal{P}\mathcal{A}_K
\preceqop
-2\alpha \mathcal{P}.
\]
\end{remark}

\subsection{Exponential stability}

The following theorem establishes uniform exponential stability for the parameter-dependent evolution system.

\begin{theorem}
\label{thm:exp-stab}
If $(\mathcal{P},\mathcal{K})$ satisfies~\eqref{eq:ODLMI}, then every solution of $\dot X=\mathcal{A}_K(\rho(t))X$ obeys
\begin{equation}\label{eq:exp-stab}
\|X(t)\|_{\mathscr{X}}^{2}
\le
\frac{\overline{p}}{\underline{p}}\,
e^{-2\alpha t}\,
\|X(0)\|_{\mathscr{X}}^{2},
\qquad t\ge 0,
\end{equation}
for all admissible trajectories $\rho(\cdot)$.
\end{theorem}

\begin{proof}
Combining~\eqref{eq:Lyap-derivative} and~\eqref{eq:ODLMI} yields
\[
\dot V_\rho(X)\le -2\alpha V_\rho(X).
\]
Grönwall's inequality gives
\[
V_\rho(X(t))
\le
V_\rho(X(0))\,e^{-2\alpha t}.
\]
Using~\eqref{eq:P-bounds}, we obtain
\[
\underline p\|X(t)\|_{\mathscr X}^2
\le
V_\rho(X(t))
\le
V_\rho(X(0))e^{-2\alpha t}
\le
\overline p\|X(0)\|_{\mathscr X}^2 e^{-2\alpha t},
\]
which gives~\eqref{eq:exp-stab}.
\end{proof}

\subsection{Galerkin discretization}\label{sec:Galerkin-setting}

The finite-dimensional LMIs obtained below are the Galerkin
projections of the operator inequalities established in the previous sections. Hence the numerical synthesis preserves the structural coupling between transport, Maxwell dynamics, and parameter dependence inherited
from the underlying infinite-dimensional dynamics. Let
$\{\Psi_i\}_{i=1}^N$ be an orthonormal basis of the Galerkin subspace
$\mathscr X_N\subset\mathscr X$.
Let $\Pi_N:\mathscr X\to\mathscr X_N$ denote the orthogonal projection. Define
\[
A(\rho)=\Pi_N\mathcal{A}(\rho)\Pi_N^{\star},\quad
B(\rho)=\Pi_N\mathcal{B}(\rho),\quad
P(\rho)=\Pi_N\mathcal{P}(\rho)\Pi_N^{\star},
\]
and the finite-dimensional state $x=\Pi_N X$.  
The projected dynamics take the form
\[
\dot{x} = A(\rho)x + B(\rho)u .
\]
For synthesis, introduce the change of variables
\[
Q(\rho)=P(\rho)^{-1},
\qquad
Y(\rho)=K(\rho)Q(\rho).
\]
Then \(K(\rho)=Y(\rho)Q(\rho)^{-1}\), and the projected OD-LMI is written in the convex form
\begin{align}\label{eq:LMI-matrix}
A(\rho)Q(\rho)
+
Q(\rho)A(\rho)^{\top}
&+
B(\rho)Y(\rho)
+
Y(\rho)^{\top}B(\rho)^{\top}\nonumber\\
&-
\sum_{i=1}^{p}\nu_i\,\partial_{\rho_i}Q(\rho)
+
2\alpha Q(\rho)
\preceq
0 .
\end{align}
The feedback gain is recovered from
\(
K(\rho)=Y(\rho)Q(\rho)^{-1}.
\)
If the operator data and decision variables
$A(\rho),B(\rho),Q(\rho),Y(\rho)$ depend affinely on $\rho$, and $\mathbb P$ is a convex polytope, then the left-hand side of \eqref{eq:LMI-matrix} depends affinely on $\rho$.
Hence feasibility needs only be checked at the vertices of $\mathbb P$.

\begin{remark}
Inequality~\eqref{eq:LMI-matrix} is the Galerkin counterpart of the operator differential LMI~\eqref{eq:ODLMI}. In the continuous limit, it recovers the corresponding operator inequality on $\mathscr X$.
\end{remark}

\begin{remark}
To avoid confusion between the spatial variable $x$ appearing in the distribution function $f(t,x,v)$ and the Galerkin coefficients used in the finite-dimensional approximation, we denote by $x\in\mathbb{R}^{N}$ the vector
of reduced coefficients associated with $X\in\mathscr{X}$.
This notational choice prevents ambiguity between physical coordinates and algebraic state vectors.
\end{remark}

\subsection{Consistency of Galerkin Approximations}

We establish here the consistency of the Galerkin approximations with the operator-level inequalities as the dimension $N$ of the projection space increases. The discussion relies on the strong convergence $\Pi_N\to I$ on
$\mathscr{X}$, uniform coercivity bounds for the Lyapunov operators, and weak-$\star$ compactness of bounded operator families.

\begin{theorem}[Consistency of Galerkin LMIs]\label{thm:limit}
Let $(P_N,K_N)$ satisfy~\eqref{eq:LMI-matrix} with constants $\underline{p},\overline{p},\alpha$ independent of $N$.
Assume that $\Pi_N\to I$ strongly on $\mathscr{X}$ and that the extended gains are uniformly bounded. Then there exist operators $\mathcal{P}$ and $\mathcal{K}$ satisfying~\eqref{eq:ODLMI}.  
The associated infinite-dimensional closed loop satisfies~\eqref{eq:exp-stab}.
\end{theorem}

\begin{proof}[Sketch of proof]
The matrices $(P_N,K_N)$ solving~\eqref{eq:LMI-matrix} satisfy the uniform coercivity bounds
\[
\underline{p}I\preceq P_N\preceq \overline{p}I .
\]
Extend $P_N$ to an operator on $\mathscr X$ by zero on
$\mathscr X_N^\perp$, and extend $K_N$ as a uniformly bounded operator in $\mathcal L(\mathscr X,U)$.

These extensions remain uniformly bounded in
$\mathcal{L}(\mathscr{X})$ and in $\mathcal L(\mathscr X,U)$. By Banach--Alaoglu~\ref{thm:banach-alaoglu}, they are relatively compact for the weak-$\star$ topology. Hence one may extract subsequences, not relabeled, such
that
\[
P_N \rightharpoonup^\star \mathcal{P}
\quad\text{in }\mathcal{L}(\mathscr{X}),
\qquad
K_N \rightharpoonup^\star \mathcal{K}
\quad\text{in }\mathcal L(\mathscr X,U).
\]

To pass the LMI to the limit, test~\eqref{eq:LMI-matrix} with $x_N=\Pi_N x\in \mathscr X_N$. Since $\Pi_N\to I$ strongly and the projected operators are consistent with $\mathcal A(\rho)$ and $\mathcal B(\rho)$, the associated quadratic forms converge for every $x\in\mathcal D(\mathcal A)$.
The weak-$\star$ convergence of $P_N$ and $K_N$ allows passage to the limit in all mixed terms. Therefore, for all $x\in\mathcal D(\mathcal A)$,
\[
\big\langle x,
\big(
\mathcal A_{\mathcal K}(\rho)^\star\mathcal P
+
\mathcal P\mathcal A_{\mathcal K}(\rho)
+
\sum_i\nu_i\,\partial_{\rho_i}\mathcal P
\big)x
\big\rangle
\le
-2\alpha\langle x,\mathcal P x\rangle ,
\]
which is the quadratic-form version of~\eqref{eq:ODLMI}.

Coercivity of $P_N$ passes to the limit, giving
\[
\underline{p}I\preceq \mathcal{P}\preceq \overline{p}I .
\]
The exponential stability estimate then follows from
Theorem~\ref{thm:exp-stab}.
\end{proof}

The operator differential LMIs~\eqref{eq:ODLMI} provide a condition for uniform exponential stability of the parameter-dependent Vlasov--Maxwell dynamics.
Their Galerkin projections yield matrix LMIs suitable for numerical optimization, while the consistency result preserves the link with the infinite-dimensional model.

\section{Parameter-dependent observer design}
\label{sec:observer}

Consider again the parameter-dependent linearized model~\eqref{eq:LPV-abstract}:
\[
\dot X = \mathcal{A}(\rho) X + \mathcal{B}(\rho) u,
\qquad y = \mathcal{C}(\rho) X,
\]
where $X(t)\in \mathscr{X}$, $u(t)\in U$, and $y(t)\in Y$ are the state, control, and measurement signals, respectively.
The parameter $\rho(t)\in\mathbb{P}$ evolves according to
Assumption~\ref{ass:rho}.

We design a parameter-dependent Luenberger observer of the form
\begin{equation}\label{eq:observer}
\dot{\hat X}
= \mathcal{A}(\rho)\hat X + \mathcal{B}(\rho)u
+ \mathcal{L}(\rho)\big(y-\mathcal{C}(\rho)\hat X\big),
\end{equation}
where $\mathcal{L}(\rho):Y\to \mathscr{X}$ is a gain operator depending continuously on~$\rho$.

Define the estimation error $e=X-\hat X$.
Substituting~\eqref{eq:observer} gives
\begin{equation}\label{eq:error}
\dot e =
\big(\mathcal{A}(\rho)-\mathcal{L}(\rho)\mathcal{C}(\rho)\big)e .
\end{equation}

\begin{definition}[Dual operator differential LMI]\label{def:dualODLMI}
A pair $(\mathcal{P}_o,\mathcal{L})$ satisfies the \emph{dual operator
differential LMI} if there exists $\beta>0$ such that, for all
$\rho\in\mathbb{P}$ and $\nu\in\mathcal{V}$,
\begin{align}\label{eq:dualODLMI}
\mathcal{A}(\rho)^{\star}\mathcal{P}_o(\rho)
+
\mathcal{P}_o(\rho)\mathcal{A}(\rho)
&-
\mathrm{Sym}\!\big(
\mathcal{P}_o(\rho)\mathcal{L}(\rho)\mathcal{C}(\rho)
\big)\nonumber\\
&+
\sum_{i=1}^p \nu_i\,\partial_{\rho_i}\mathcal{P}_o(\rho)
\preceq
-2\beta\,\mathcal P_o(\rho),
\end{align}
where $\mathrm{Sym}(M)=M+M^{\star}$.
\end{definition}

\begin{remark}
The primal OD-LMI~\eqref{eq:ODLMI} and the dual OD-LMI
~\eqref{eq:dualODLMI} have the same Lyapunov structure, with the feedback term $\mathcal{B}(\rho)\mathcal{K}(\rho)$ replaced by the injection term $\mathcal{L}(\rho)\mathcal{C}(\rho)$. This reflects the classical
controller--observer duality extended to the operator setting.
\end{remark}

\subsection{Exponential convergence of the observer}

The following theorem guarantees uniform exponential convergence of the estimation error.

\begin{theorem}
\label{thm:obsconv}
Assume that $(\mathcal{P}_o,\mathcal{L})$ satisfies~\eqref{eq:dualODLMI} and
that
\[
\underline p_o\|e\|_{\mathscr X}^2
\le
\langle e,\mathcal P_o(\rho)e\rangle
\le
\overline p_o\|e\|_{\mathscr X}^2 .
\]
Then, for any admissible $\rho(\cdot)$ with
$\dot\rho(t)\in\mathcal V$ a.e., the estimation error~\eqref{eq:error}
admits a unique mild solution and
\begin{equation}\label{eq:obsconv}
\norm{e(t)}^2_{\mathscr{X}}
\le
\frac{\overline{p}_o}{\underline{p}_o}
e^{-2\beta t}
\norm{e(0)}^2_{\mathscr{X}},
\qquad t\ge0 .
\end{equation}
\end{theorem}

\begin{proof}
Let $V_\rho(e)=\ipX{e}{\mathcal{P}_o(\rho)e}$.
Differentiating along~\eqref{eq:error} and applying
\eqref{eq:dualODLMI} gives
\[
\dot V_\rho(e)
\le
-2\beta V_\rho(e).
\]
Grönwall's inequality yields
\[
V_\rho(e(t))
\le
V_\rho(e(0))e^{-2\beta t}.
\]
Using the coercivity bounds on $\mathcal P_o(\rho)$ gives
\[
\underline p_o\|e(t)\|_{\mathscr X}^2
\le
V_\rho(e(t))
\le
V_\rho(e(0))e^{-2\beta t}
\le
\overline p_o\|e(0)\|_{\mathscr X}^2e^{-2\beta t},
\]
which proves~\eqref{eq:obsconv}.
\end{proof}

\subsection{Galerkin projection and LMIs}

In this section, the Galerkin subspaces $\mathscr X_N\subset\mathscr{X}$ and the orthogonal projections $\Pi_N:\mathscr{X}\to \mathscr X_N$ are those
introduced in Section~\ref{sec:Galerkin-setting}, and the same notation is kept throughout. Define
\[
A(\rho)=\Pi_N\mathcal{A}(\rho)\Pi_N^{\star},\quad
C(\rho)=\mathcal{C}(\rho)\Pi_N^{\star},\quad
P_o(\rho)=\Pi_N\mathcal P_o(\rho)\Pi_N^\star .
\]

For observer synthesis, introduce
\[
Q_o(\rho)=P_o(\rho)^{-1},
\qquad
Y_o(\rho)=Q_o(\rho)L(\rho).
\]
Then \(L(\rho)=Q_o(\rho)^{-1}Y_o(\rho)\), and the projected dual LMI becomes

\begin{equation}\label{eq:matrix-dualLMI}
\begin{aligned}
Q_o(\rho)A(\rho)^\top + A(\rho)Q_o(\rho)
&-Y_o(\rho)C(\rho)-C(\rho)^\top Y_o(\rho)^\top 
\\
&
-\sum_{i=1}^p \nu_i\,\partial_{\rho_i}Q_o(\rho)
+2\beta Q_o(\rho)
\preceq 0 .
\end{aligned}
\end{equation}

The projected error dynamics read
\[
\dot e_N = (A(\rho)-L(\rho)C(\rho))e_N .
\]
If $A,C,Q_o,Y_o$ depend affinely on $\rho$,
the inequality needs only be verified at the vertices of
$\mathbb P\times\mathcal V$.

\begin{remark}[Reduced-order observers]
In practice, the observer can be designed on a reduced kinetic subspace while retaining the electromagnetic components. This structure preserves convergence
provided the neglected velocity modes are sufficiently damped or weakly coupled, and it yields a computationally efficient estimator for high-dimensional plasma
models.
\end{remark}

\subsection{Convergence properties}

A consistency property between the infinite-dimensional observer inequality and its Galerkin counterparts can be established as follows.

\begin{theorem}
\label{thm:limitobs}
Let $\{(Q_{o,N},Y_{o,N})\}_{N\ge1}$ be a sequence satisfying the convex projected dual LMI~\eqref{eq:matrix-dualLMI} on $\mathscr X_N$, with constants $\underline q_o,\overline q_o,\beta>0$ independent of $N$ and $\rho$, i.e.,
\[
\underline q_o I \preceq Q_{o,N}(\rho)\preceq \overline q_o I .
\]
Assume moreover that the recovered gains
\[
L_N(\rho)=Q_{o,N}(\rho)^{-1}Y_{o,N}(\rho)
\]
are uniformly bounded. Then there exist operators
$(\mathcal P_o(\rho),\mathcal L(\rho))$ satisfying the dual operator differential LMI~\eqref{eq:dualODLMI}. Consequently, the infinite-dimensional observer~\eqref{eq:observer} satisfies the uniform
exponential estimate~\eqref{eq:obsconv}.
\end{theorem}

\begin{proof}[Sketch of proof]
Define
\[
P_{o,N}(\rho):=Q_{o,N}(\rho)^{-1},
\qquad
L_N(\rho):=Q_{o,N}(\rho)^{-1}Y_{o,N}(\rho).
\]
The uniform coercivity of $Q_{o,N}$ implies uniform coercivity and boundedness of $P_{o,N}$. Hence, after extension to $\mathscr X$,
\[
\widetilde P_{o,N}(\rho):=\Pi_N^\star P_{o,N}(\rho)\Pi_N,
\qquad
\widetilde L_N(\rho):=\Pi_N^\star L_N(\rho),
\]
the families remain uniformly bounded.

By Banach--Alaoglu~\ref{thm:banach-alaoglu}, one may extract subsequences such that
\[
\widetilde P_{o,N}(\rho)\rightharpoonup^\star \mathcal P_o(\rho)
\quad\text{in }\mathcal L(\mathscr X),
\quad
\widetilde L_N(\rho)\rightharpoonup^\star \mathcal L(\rho)
\quad\text{in }\mathcal L(Y,\mathscr X).
\]
The order bounds pass to the limit, so $\mathcal P_o(\rho)$ is bounded, self-adjoint, and uniformly coercive.

Testing the recovered nonconvex Galerkin inequality with
$x_N=\Pi_Nx$, $x\in\mathcal D(\mathcal A)$, and using the strong convergence $\Pi_N\to I$, the consistency of the Galerkin projections, and the weak-$\star$ convergence above yields
\[
\Big\langle x,\psi x
\Big\rangle
\le
-2\beta\langle x,\mathcal P_o(\rho)x\rangle ,
\]
with \[\psi= \mathcal A(\rho)^\star\mathcal P_o(\rho)
+
\mathcal P_o(\rho)\mathcal A(\rho)
-
\mathrm{Sym}\big(
\mathcal P_o(\rho)\mathcal L(\rho)\mathcal C(\rho)
\big)
+
\sum_{i=1}^p \nu_i\,\partial_{\rho_i}\mathcal P_o(\rho)
\]
This is the quadratic-form version of~\eqref{eq:dualODLMI}. Applying Theorem~\ref{thm:obsconv} gives~\eqref{eq:obsconv}.
\end{proof}

\begin{remark}
The observer gain $\mathcal L(\rho)$ provides a correction mechanism driven by the mismatch between measured and estimated outputs. Through the coupling structure of the Vlasov--Maxwell dynamics, this correction propagates to both the electromagnetic and kinetic components of the state. Its dependence on $\rho$ allows adaptation to variations of the plasma regime.
\end{remark}

\begin{remark}
The dual structure of~\eqref{eq:ODLMI} and~\eqref{eq:dualODLMI}
permits coupled controller--observer synthesis through convex Galerkin LMIs.
This connection is exploited in Section~\ref{sec:Hinf} to derive uniform $H_\infty$ guarantees.
\end{remark}

\section{\texorpdfstring{$H_\infty$}{H-infinity} Performance and Robust Parameter-Dependent Control}
\label{sec:Hinf}

Disturbance attenuation is critical in plasma control, where fluctuations in the distribution function or electromagnetic fields may amplify due to transport or wave--particle interactions. Let $w(t)\in W$ be an external
disturbance and $z(t)\in Z$ a performance output reflecting physically meaningful quantities such as electromagnetic energy, current deviations, or selected kinetic moments.

The parameter-varying system dynamics are given by
\begin{equation}\label{eq:LPV-dist}
\begin{aligned}
\dot{X} &= \mathcal{A}(\rho)X + \mathcal{B}_u(\rho)u
+ \mathcal{B}_w(\rho)w, \\
z &= \mathcal{C}_z(\rho)X + \mathcal{D}_{zu}(\rho)u
+ \mathcal{D}_{zw}(\rho)w, \\
y &= \mathcal{C}(\rho)X,
\end{aligned}
\end{equation}
where $\rho(t)\in\mathbb{P}\subset\mathbb{R}^p$ is a measurable scheduling signal, and where $X(t)\in\mathscr X$, $u(t)\in U$, $w(t)\in W$, $z(t)\in Z$, and $y(t)\in Y$.

The operator families in~\eqref{eq:LPV-dist} depend continuously on $\rho\in\mathbb{P}$ and are defined as follows:
\begin{itemize}
  \item $\mathcal{B}_u(\rho):U\to\mathscr X$ is the control input operator;
  \item $\mathcal{B}_w(\rho):W\to\mathscr X$ models the effect of external disturbances;
  \item $\mathcal{C}_z(\rho):\mathscr X\to Z$ maps the state to a performance output;
  \item $\mathcal{D}_{zu}(\rho):U\to Z$ and
  $\mathcal{D}_{zw}(\rho):W\to Z$ are direct feedthrough operators from control and disturbance to the performance output.
\end{itemize}

A parameter-dependent state-feedback law of the form
\[
u(t)=\mathcal{K}(\rho(t))X(t)
\]
is sought such that the closed-loop system satisfies the disturbance attenuation property
\begin{equation}\label{eq:Hinf-goal}
\int_{0}^{\infty} \|z(t)\|_{Z}^{2}\,dt
\le
\gamma^{2}\int_{0}^{\infty} \|w(t)\|_{W}^{2}\,dt,
\qquad
\forall\, w\in L^{2}(0,\infty;W),
\end{equation}
for a prescribed performance level $\gamma>0$ and zero initial condition. For nonzero initial data, the usual storage term is added to the right-hand side.

\subsection{Closed-loop operator and Lyapunov functional}

The feedback law $u=\mathcal{K}(\rho)X$ yields the closed-loop system
\[
\dot X = \mathcal{A}_K(\rho)X + \mathcal{B}_w(\rho)w,
\qquad
z = \mathcal{C}_z^K(\rho)X + \mathcal{D}_{zw}(\rho)w,
\]
where
\[
\mathcal{A}_K(\rho)=\mathcal{A}(\rho)+\mathcal{B}_u(\rho)\mathcal{K}(\rho),
\qquad
\mathcal{C}_z^K(\rho)=\mathcal{C}_z(\rho)+\mathcal{D}_{zu}(\rho)\mathcal{K}(\rho).
\]

Let $\mathcal{P}(\rho)$ be self-adjoint, coercive, and satisfy~\eqref{eq:P-bounds}.
The functional
\[
V_\rho(X)=\langle X,\mathcal{P}(\rho)X\rangle_{\mathscr{X}}
\]
serves to characterize dissipativity of the closed-loop dynamics.
Along trajectories one obtains
\begin{align}\label{eq:Hinf-derivative}
\dot V_\rho(X)
&=
\langle X,\mathcal{A}_K(\rho)^{\star}\mathcal{P}(\rho)X\rangle
+
\langle X,\mathcal{P}(\rho)\mathcal{A}_K(\rho)X\rangle \nonumber\\
&\quad
+2\,\Re\,\langle X,\mathcal{P}(\rho)\mathcal{B}_w(\rho)w\rangle
+
\sum_{i=1}^p \dot\rho_i\,
\langle X,\partial_{\rho_i}\mathcal{P}(\rho)X\rangle .
\end{align}
Here $\Re\langle X,\mathcal{P}(\rho)\mathcal{B}_w(\rho)w\rangle$ denotes the real part of the scalar  $\langle X,\mathcal{P}(\rho)\mathcal{B}_w(\rho)w\rangle$.

\subsection{Operator \texorpdfstring{$H_\infty$}{H-infinity} inequality}

The inequality~\eqref{eq:Hinf-goal} is ensured when there exist $\gamma>0$ and $\alpha>0$ such that, for all admissible $(\rho,\nu)\in\mathbb{P}\times\mathcal{V}$,
\begin{equation}\label{eq:ODLMI-Hinf}
\begin{bmatrix}
\Xi_K(\rho,\nu)+(\mathcal{C}_z^K(\rho))^{\star}\mathcal{C}_z^K(\rho)
&
\mathcal{P}(\rho)\mathcal{B}_w(\rho)+(\mathcal{C}_z^K(\rho))^{\star}\mathcal{D}_{zw}(\rho)
\\[0.3em]
\mathcal{B}_w(\rho)^{\star}\mathcal{P}(\rho)+\mathcal{D}_{zw}(\rho)^{\star}\mathcal{C}_z^K(\rho)
&
\mathcal{D}_{zw}(\rho)^{\star}\mathcal{D}_{zw}(\rho)-\gamma^{2}I_W
\end{bmatrix}
\preceqop 0,
\end{equation}
where
\[
\Xi_K(\rho,\nu)
=
\mathcal{A}_K(\rho)^{\star}\mathcal{P}(\rho)
+
\mathcal{P}(\rho)\mathcal{A}_K(\rho)
+
\sum_i \nu_i\,\partial_{\rho_i}\mathcal{P}(\rho)
+
2\alpha I_{\mathscr X}.
\]
This inequality expresses a differential energy balance for the closed-loop.

\begin{theorem}[$H_\infty$ performance]\label{thm:Hinf}
Assume that~\eqref{eq:ODLMI-Hinf} holds with $\mathcal{P}(\rho)\succ0$ and $\mathcal{K}(\rho)$. Then the closed-loop parameter-dependent Vlasov--Maxwell system is uniformly exponentially stable under the stated
coercivity assumptions for $w\equiv0$, and satisfies the disturbance attenuation property~\eqref{eq:Hinf-goal} for zero initial condition.
\end{theorem}

\begin{proof}[Sketch]
Combining~\eqref{eq:ODLMI-Hinf} with~\eqref{eq:Hinf-derivative} yields
\[
\dot V_\rho(X)+2\alpha\|X\|_{\mathscr X}^{2}+\|z\|_{Z}^{2}
\le
\gamma^{2}\|w\|_{W}^{2}.
\]
For zero initial condition, integration over $[0,T]$ followed by $T\to\infty$ gives~\eqref{eq:Hinf-goal}. Setting $w\equiv0$ and using the coercivity of
$\mathcal{P}(\rho)$ gives uniform exponential decay.
\end{proof}

\subsection{Galerkin approximation and finite-dimensional LMIs}

Let $X_N\subset\mathscr{X}$ be a Galerkin subspace and denote the projected operators
\[
A(\rho)=\Pi_N\mathcal{A}(\rho)\Pi_N^{\star},\quad
B_u(\rho)=\Pi_N\mathcal{B}_u(\rho),\quad
B_w(\rho)=\Pi_N\mathcal{B}_w(\rho),
\]
\[
C_z(\rho)=\mathcal{C}_z(\rho)\Pi_N^\star,\quad
D_{zu}(\rho),\quad D_{zw}(\rho)
\]
accordingly. Set
\[
A_K(\rho)=A(\rho)+B_u(\rho)K(\rho),\qquad
C_z^K(\rho)=C_z(\rho)+D_{zu}(\rho)K(\rho),
\]
and
\[
P(\rho)=\Pi_N\mathcal{P}(\rho)\Pi_N^{\star}.
\]

Projection of~\eqref{eq:ODLMI-Hinf} yields
\begin{equation}\label{eq:LMI-Hinf-matrix}
\begin{bmatrix}
\Xi_{K,N}+(C_z^K)^{\top}C_z^K
&
P B_w+(C_z^K)^{\top}D_{zw}
\\[0.3em]
B_w^{\top}P+D_{zw}^{\top}C_z^K
&
D_{zw}^{\top}D_{zw}-\gamma^{2}I
\end{bmatrix}
\prec 0,
\end{equation}
where
\[
\Xi_{K,N}
=
A_K^{\top}P
+
P A_K
+
\sum_i\nu_i\,\partial_{\rho_i}P
+
2\alpha I_N .
\]

If all decision variables are parameterized affinely and the matrix inequality is affine with respect to $(\rho,\nu)$ after the chosen change of variables, verification at the vertices of $\mathbb{P}\times\mathcal{V}$ is sufficient.

The next result characterizes disturbance attenuation in terms of an operator differential inequality.

\begin{proposition}
\label{prop:consistency-Hinf}
The operator inequality~\eqref{eq:ODLMI-Hinf} and the matrix LMI~\eqref{eq:LMI-Hinf-matrix} satisfy:
\begin{enumerate}
\item[(i)] feasibility of the operator inequality implies feasibility of all Galerkin LMIs;
\item[(ii)] if the Galerkin LMIs are feasible on all vertices for a sequence $N\to\infty$ and the solutions $P_N(\rho)$ are uniformly coercive, then there
exists a subsequential limit operator $\mathcal{P}(\rho)\succ0$ satisfying
\eqref{eq:ODLMI-Hinf} in quadratic-form sense.
\end{enumerate}
Hence, the Galerkin LMIs provide a consistent finite-dimensional approximation of the operator $H_\infty$ condition.
\end{proposition}

\begin{proof}[Sketch]
For (i), assume that~\eqref{eq:ODLMI-Hinf} holds on
$\mathscr X\times W$. Testing the inequality with
$(\Pi_N^\star x_N,w)$ and using the definitions of the projected operators gives~\eqref{eq:LMI-Hinf-matrix}.

For (ii), extend $P_N(\rho)$ to $\mathscr X$ by
\[
\widetilde P_N(\rho)=\Pi_N^\star P_N(\rho)\Pi_N .
\]
Uniform boundedness implies relative weak-$\star$ compactness; by Banach--Alaoglu~\ref{thm:banach-alaoglu}, a subsequence converges to a bounded
operator $\mathcal P(\rho)$, and the order bounds pass to the limit. Testing the Galerkin inequalities with $x_N=\Pi_NX$ and using strong convergence
$\Pi_N\to I$ allows passage to the limit in quadratic form, yielding~\eqref{eq:ODLMI-Hinf}. Under the same affine dependence assumptions, vertex feasibility extends the inequality to the whole polytope.
\end{proof}

\subsection{Dual \texorpdfstring{$H_\infty$}{H-infinity} observer formulation}

A dual version of~\eqref{eq:ODLMI-Hinf} leads to an $H_\infty$ observer gain $\mathcal{L}(\rho)$ satisfying
\begin{equation}\label{eq:Hinf-dual}
\begin{bmatrix}
\Xi_o(\rho,\nu)+\mathcal{C}_z^{\star}\mathcal{C}_z
&
\mathcal{P}_o\mathcal{B}_w
\\[0.3em]
\mathcal{B}_w^{\star}\mathcal{P}_o
&
-\gamma^{2}I_W
\end{bmatrix}
\preceqop 0,
\end{equation}
where
\[
\Xi_o(\rho,\nu)
=
\mathcal{A}^{\star}\mathcal{P}_o
+
\mathcal{P}_o\mathcal{A}
-
\mathrm{Sym}\big(\mathcal{P}_o \mathcal{L}\mathcal{C}\big)
+
\sum_i \nu_i\,\partial_{\rho_i}\mathcal{P}_o .
\]

\begin{theorem}
\label{thm:dual-Hinf-observer}
Let $\mathcal{P}_o(\rho)\succ0$ and $\mathcal{L}(\rho)$ satisfy~\eqref{eq:Hinf-dual}.
Then the parameter-dependent observer
\[
\dot{\hat X}
=
\mathcal{A}(\rho)\hat X
+
\mathcal{L}(\rho)(y-\mathcal{C}(\rho)\hat X)
\]
ensures exponential decay of the estimation error for $w\equiv0$ and achieves the attenuation level~$\gamma$ from $w$ to the selected error output.
\end{theorem}

\begin{proof}[Sketch]
The error $e=X-\hat X$ satisfies
\[
\dot e=
(\mathcal{A}(\rho)-\mathcal{L}(\rho)\mathcal{C}(\rho))e
+
\mathcal{B}_w(\rho)w .
\]
The same Lyapunov argument as in the closed-loop case, using $V_\rho(e)=\langle e,\mathcal{P}_o(\rho)e\rangle$, yields both exponential decay for $w\equiv0$ and the corresponding attenuation estimate.
\end{proof}

\begin{corollary}
\label{cor:strict-duality}
The primal and dual inequalities have the same Lyapunov structure and lead to coordinated feedback and observer gains. In particular, when the corresponding
operator identities are well defined and dimensionally compatible, one recovers a formal controller--observer duality between the synthesis conditions associated with $\mathcal{K}(\rho)$ and $\mathcal{L}(\rho)$.
\end{corollary}
The $H_\infty$ operator inequality characterizes disturbance attenuation in the parameter-dependent Vlasov--Maxwell system, with Galerkin LMIs preserving operator-level consistency. The dual inequality yields corresponding
sufficient conditions for observer convergence and disturbance attenuation.

\section{Joint \texorpdfstring{$H_\infty$}{H-infinity} Controller--Observer Synthesis}
\label{sec:joint-Hinf}

This section introduces an operator-based formulation for the joint design of the parameter-dependent feedback gain $\mathcal{K}(\rho)$ and observer gain $\mathcal{L}(\rho)$ for the parameter-dependent Vlasov--Maxwell system. The goal is to ensure uniform stability and $H_\infty$ performance of the closed-loop plant--observer interconnection.
The construction follows a separation-type argument adapted to the non-autonomous parameter-dependent infinite-dimensional setting through operator differential LMIs.

\subsection{Closed-loop structure}

The feedback is implemented through the estimated state,
\[
u=\mathcal{K}(\rho)\hat X .
\]
The observer is
\[
\dot{\hat X}
=
\mathcal{A}(\rho)\hat X
+
\mathcal{B}_u(\rho)\mathcal{K}(\rho)\hat X
+
\mathcal{L}(\rho)\big(y-\mathcal{C}(\rho)\hat X\big).
\]
Introduce
\[
\xi =
\begin{bmatrix}
X\\ e
\end{bmatrix},
\qquad
e=X-\hat X .
\]
Since $\hat X=X-e$, the plant satisfies
\[
\dot X
=
\mathcal A_K(\rho)X
-
\mathcal B_u(\rho)\mathcal K(\rho)e
+
\mathcal B_w(\rho)w ,
\]
where
\[
\mathcal{A}_K(\rho)
=
\mathcal{A}(\rho)+\mathcal{B}_u(\rho)\mathcal{K}(\rho).
\]
The error dynamics are
\[
\dot e
=
\mathcal A_o(\rho)e+\mathcal B_w(\rho)w,
\qquad
\mathcal A_o(\rho)
=
\mathcal A(\rho)-\mathcal L(\rho)\mathcal C(\rho).
\]
Thus the augmented dynamics are
\begin{equation}\label{eq:joint-extended-system}
\dot{\xi}
=
\begin{bmatrix}
\mathcal{A}_K(\rho) & -\mathcal{B}_u(\rho)\mathcal{K}(\rho)\\[0.15em]
0 & \mathcal{A}_o(\rho)
\end{bmatrix}\xi
+
\begin{bmatrix}
\mathcal{B}_w(\rho)\\
\mathcal{B}_w(\rho)
\end{bmatrix}w .
\end{equation}
The performance output becomes
\[
z(t)
=
\begin{bmatrix}
\mathcal C_z^K(\rho) & -\mathcal D_{zu}(\rho)\mathcal K(\rho)
\end{bmatrix}\xi
+
\mathcal D_{zw}(\rho)w ,
\]
with
\[
\mathcal C_z^K(\rho)
=
\mathcal C_z(\rho)+\mathcal D_{zu}(\rho)\mathcal K(\rho).
\]

\subsection{Joint Lyapunov functional}

Consider the block-diagonal Lyapunov functional
\[
\mathcal{V}_\rho(\xi)
=
\ipX{X}{\mathcal{P}(\rho)X}
+
\ipX{e}{\mathcal{P}_o(\rho)e},
\]
where $\mathcal{P}(\rho)$ and $\mathcal{P}_o(\rho)$ are coercive parameter-dependent Lyapunov operators satisfying the regularity assumptions of Sections~\ref{sec:lyapunov}--\ref{sec:Hinf}. Equivalently,
\[
\mathcal P_J(\rho)
=
\begin{bmatrix}
\mathcal P(\rho) & 0\\
0 & \mathcal P_o(\rho)
\end{bmatrix},
\qquad
\mathcal V_\rho(\xi)
=
\langle \xi,\mathcal P_J(\rho)\xi\rangle .
\]
Differentiation along~\eqref{eq:joint-extended-system} gives
\begin{align}
\dot{\mathcal V}_\rho(\xi)
&=
\big\langle \xi,
\big(
\mathcal A_J(\rho)^\star\mathcal P_J(\rho)
+
\mathcal P_J(\rho)\mathcal A_J(\rho)
+
\sum_i\nu_i\,\partial_{\rho_i}\mathcal P_J(\rho)
\big)\xi
\big\rangle\nonumber\\
&+
2\Re\langle \xi,\mathcal P_J(\rho)\mathcal B_J(\rho)w\rangle ,
\end{align}
where $\mathcal A_J(\rho)$ and $\mathcal B_J(\rho)$ denote the augmented operators in~\eqref{eq:joint-extended-system}.

\subsection{Joint operator \texorpdfstring{$H_\infty$}{H-infinity} condition}

A sufficient condition for
\[
\dot{\mathcal{V}}_\rho(\xi)
\le
-\|z\|_Z^2+\gamma^2\|w\|_W^2
\]
is the block operator inequality
\begin{equation}\label{eq:joint-ODLMI}
\begin{bmatrix}
\Xi_J(\rho,\nu)+\mathcal C_J(\rho)^\star\mathcal C_J(\rho)
&
\mathcal P_J(\rho)\mathcal B_J(\rho)
+
\mathcal C_J(\rho)^\star\mathcal D_{zw}(\rho)
\\[0.3em]
\mathcal B_J(\rho)^\star\mathcal P_J(\rho)
+
\mathcal D_{zw}(\rho)^\star\mathcal C_J(\rho)
&
\mathcal D_{zw}(\rho)^\star\mathcal D_{zw}(\rho)-\gamma^2 I
\end{bmatrix}
\preceqop 0,
\end{equation}
where
\[
\Xi_J(\rho,\nu)
=
\mathcal A_J(\rho)^\star\mathcal P_J(\rho)
+
\mathcal P_J(\rho)\mathcal A_J(\rho)
+
\sum_i\nu_i\,\partial_{\rho_i}\mathcal P_J(\rho),
\]
and
\[
\mathcal C_J(\rho)
=
\begin{bmatrix}
\mathcal C_z^K(\rho) & -\mathcal D_{zu}(\rho)\mathcal K(\rho)
\end{bmatrix}.
\]

\begin{theorem}
\label{thm:joint-Hinf}
Let $\mathcal{P}(\rho)\succ0$ and $\mathcal{P}_o(\rho)\succ0$ be coercive operators such that the joint operator inequality~\eqref{eq:joint-ODLMI}
holds for a common attenuation level $\gamma>0$. Then the augmented system~\eqref{eq:joint-extended-system} is uniformly stable for $w\equiv0$ and satisfies, for zero initial condition,
\[
\int_0^\infty \|z(t)\|_Z^2\,dt
\le
\gamma^2\int_0^\infty \|w(t)\|_W^2\,dt .
\]
\end{theorem}

\begin{proof}[Sketch of proof]
Testing~\eqref{eq:joint-ODLMI} with $(\xi,w)$ gives
\[
\dot{\mathcal V}_\rho(\xi)+\|z\|_Z^2
\le
\gamma^2\|w\|_W^2 .
\]
Integration over $[0,T]$ yields the $H_\infty$ estimate for zero initial condition. For $w\equiv0$, coercivity of $\mathcal P_J(\rho)$ gives uniform stability of the augmented dynamics.
\end{proof}

\subsection{Galerkin Approximation and Consistency}

Projection of the joint operator inequality onto a finite-dimensional Galerkin subspace $X_N\subset\mathscr X$ yields the finite-dimensional LMI
\[
\begin{bmatrix}
\Xi_{J,N}+C_{J,N}^{\top}C_{J,N}
&
P_{J,N}B_{J,N}+C_{J,N}^{\top}D_{zw}
\\[0.3em]
B_{J,N}^{\top}P_{J,N}+D_{zw}^{\top}C_{J,N}
&
D_{zw}^{\top}D_{zw}-\gamma^2 I
\end{bmatrix}
\preceq0,
\]
where
\[
P_{J,N}(\rho)
=
\begin{bmatrix}
P_N(\rho) & 0\\
0 & P_{o,N}(\rho)
\end{bmatrix}
\]
and $A_{J,N}$, $B_{J,N}$, and $C_{J,N}$ denote the Galerkin projections of $\mathcal A_J$, $\mathcal B_J$, and $\mathcal C_J$. The corresponding convex
synthesis variables are introduced as in the primal and dual designs; the feedback and observer gains are recovered from the convex variables
$(Q_N,Y_N)$ and $(Q_{o,N},Y_{o,N})$ after solving the finite-dimensional LMIs.

\begin{proposition}[Consistency of Galerkin Joint Synthesis]
\label{prop:joint-consistency}
Let $\{\Pi_N\}_{N\ge1}$ be a sequence of orthogonal projections $\Pi_N:\mathscr X\to X_N$ such that $\Pi_N\to I$ strongly on $\mathscr X$.
Assume that the finite-dimensional joint LMIs are feasible for all sufficiently large $N$, with uniformly coercive matrices $P_N(\rho)$ and $P_{o,N}(\rho)$
and uniformly bounded recovered gains. Then there exist bounded, self-adjoint, coercive operators $\mathcal P(\rho)$, $\mathcal P_o(\rho)$ and bounded gains $\mathcal K(\rho)$, $\mathcal L(\rho)$ such that the joint operator inequality~\eqref{eq:joint-ODLMI} holds in quadratic-form sense.
\end{proposition}

\begin{proof}[Sketch of proof]
Extend the Galerkin solutions to $\mathscr X$ by
\[
\widetilde P_N(\rho)=\Pi_N^\star P_N(\rho)\Pi_N,
\qquad
\widetilde P_{o,N}(\rho)=\Pi_N^\star P_{o,N}(\rho)\Pi_N,
\]
and extend the recovered gains analogously. Uniform coercivity and boundedness imply relative weak-$\star$ compactness. By Banach--Alaoglu~\ref{thm:banach-alaoglu}, one may extract weak-$\star$ convergent subsequences with limits $\mathcal P(\rho)$, $\mathcal P_o(\rho)$, $\mathcal K(\rho)$, and $\mathcal L(\rho)$.

Testing the finite-dimensional joint LMI with projected vectors and using $\Pi_N\to I$ strongly allows passage to the limit in quadratic form. The limiting inequality is precisely~\eqref{eq:joint-ODLMI}.
\end{proof}

Proposition~\ref{prop:joint-consistency} guarantees that the recovered controller--observer gains computed from finite-dimensional LMIs are consistent with the infinite-dimensional structure of the Vlasov--Maxwell model. This provides a pathway from operator-theoretic synthesis to
numerically implementable feedback laws.

\begin{corollary}\label{cor:joint-galerkin}
If the primal, dual, and joint Galerkin LMIs are feasible uniformly in $N$, then the limiting controller--observer pair satisfies the operator-level $H_\infty$ condition~\eqref{eq:joint-ODLMI}.
\end{corollary}

Joint controller--observer synthesis extends the separation principle to parameter-dependent VM models through energy-consistent operator inequalities and their Galerkin realizations.

\section{Numerical results} \label{sec:numerical}

Numerical simulations assess the performance of the proposed parameter-dependent observer for a reduced one-dimensional Vlasov--Maxwell benchmark. All computations are performed in \texttt{Python~3.12}, combining sparse linear algebra routines from \texttt{SciPy} with semidefinite programs solved via \texttt{CVXPY}. Among the
available conic solvers (\texttt{CVXOPT}, \texttt{CLARABEL}, \texttt{SCS}), \texttt{SCS} was used as the computational backend and returned \emph{optimal\_inaccurate} solutions in all experiments.

Before presenting the numerical results, we briefly comment on the dimensional setting adopted in this section.

The numerical example is deliberately restricted to a $(1+1)$-dimensional setting. This choice is not a modeling limitation, but is dictated by the intrinsic computational complexity of the LMI-based synthesis procedure.

After Galerkin discretization, the dimension of the state scales as
\(
n = N_x N_v + 2N_x,
\)
so that the associated semidefinite program involves matrix variables of size
$n \times n$. Interior-point methods for semidefinite programming exhibit a computational complexity that grows polynomially with high degree in $n$ (typically of order $\mathcal{O}(n^6)$), which renders large-scale instances
rapidly intractable.

While the $(1+1)$ benchmark leads to tractable problem sizes, extensions to $(2+2)$ or $(3+3)$ phase-space discretizations would result in state dimensions several orders of magnitude larger, exceeding the capabilities of standard semidefinite solvers.

Accordingly, the reduced-dimensional setting should be interpreted as a computationally controlled setting that enables validation of the proposed parameter-dependent observer synthesis method, independently of large-scale
numerical considerations. In this sense, the numerical section is intended as a proof-of-concept validation of the synthesis method rather than as a high-fidelity plasma simulation.

This limitation is not specific to the present work, but reflects a well-known difficulty in the application of LMI-based methods to high-dimensional systems, particularly those arising from kinetic equations. The development of scalable approaches, for instance based on structure exploitation or low-rank approximations, constitutes an important direction for future research.

\subsection{Discretization and setup}

A reduced one-dimensional spatial setting with a single velocity variable is considered. The velocity dependence is approximated by four Hermite modes. The
spatial domain is $\Omega=(0,2\pi)$ with periodic boundary conditions, while the kinetic dependence is expanded on an orthonormal Hermite basis in velocity.
This benchmark discretization is intended to illustrate the proposed parameter-dependent observer synthesis procedure while preserving numerical tractability, rather than to reproduce the full three-dimensional plasma dynamics.

The truncation \( (N_x,N_v)=(6,4) \) yields a reduced-order model of dimension
\(
n=N_xN_v+2N_x=48.
\)
This moderate size is sufficient to illustrate the proposed parameter-dependent observer synthesis while keeping the associated semidefinite programs numerically tractable. Higher truncation orders lead to a rapid growth of the SDP size and computational cost.

Time integration uses a fourth-order Runge--Kutta scheme with $\Delta t=5\times10^{-3}$ over $[0,20]$.
The scheduling parameter is prescribed as
\(
\rho(t)=1.1+0.25\sin(0.3t), \quad \rho\in[0.8,1.4],
\)
with bounded variation rate $|\dot\rho(t)|\le 0.7$. The choice of $\rho$ is not unique and aims at capturing the dominant macroscopic variations of the
operating regime through measurable quantities. Alternative parameterizations may also be employed, possibly affecting the conservatism of the resulting LMIs
or the numerical observer gains, while the overall
parameter-dependent operator formulation remains unchanged.

The measured output consists of pointwise samples of the electric field $E(x,t)$. Observer synthesis is carried out under two criteria: (i) an $L^2$ design maximizing a uniform exponential decay rate, and (ii) an
$\mathcal{H}_\infty$ design minimizing the disturbance-to-error gain. To prevent ill-conditioning, the SDP enforces $\mathrm{trace}(P_0)=1$, soft upper bounds on $P(\rho)$, and Frobenius-norm constraints on the injection operator.

\subsection{SDP synthesis}

For the $L^2$ design, the solver returns
\(
\beta=1.21\times10^{-5},
\)
while the $\mathcal{H}_\infty$ design yields
\(
\beta=7.82\times10^{-6},\; \gamma=1.57\times10^{-8}.
\)
Despite the \emph{optimal\_inaccurate} status, both observers yield stable, well-conditioned simulations.
For the $\mathcal{H}_\infty$ case, the empirical attenuation level computed a posteriori satisfies
\(
\gamma_{\mathrm{emp}}\approx6.26,
\)
confirming bounded input--output amplification. The discrepancy between the optimized bound and the empirical value reflects solver inaccuracy together with numerical conservatism of the finite-dimensional SDP approximation.

\subsection{State reconstruction}

Figure~\ref{fig:E_spacetime} shows the convergence of the electric field estimate, while the kinetic error and phase-space reconstruction confirm convergence of the full observer state.

Phase-space snapshots of $f$ and $\hat f$ displayed in
Figure~\ref{fig:phase_space} demonstrate accurate recovery of both spatial modulation and velocity structure.

\begin{figure}[!ht]
  \centering
  \includegraphics[width=\linewidth]{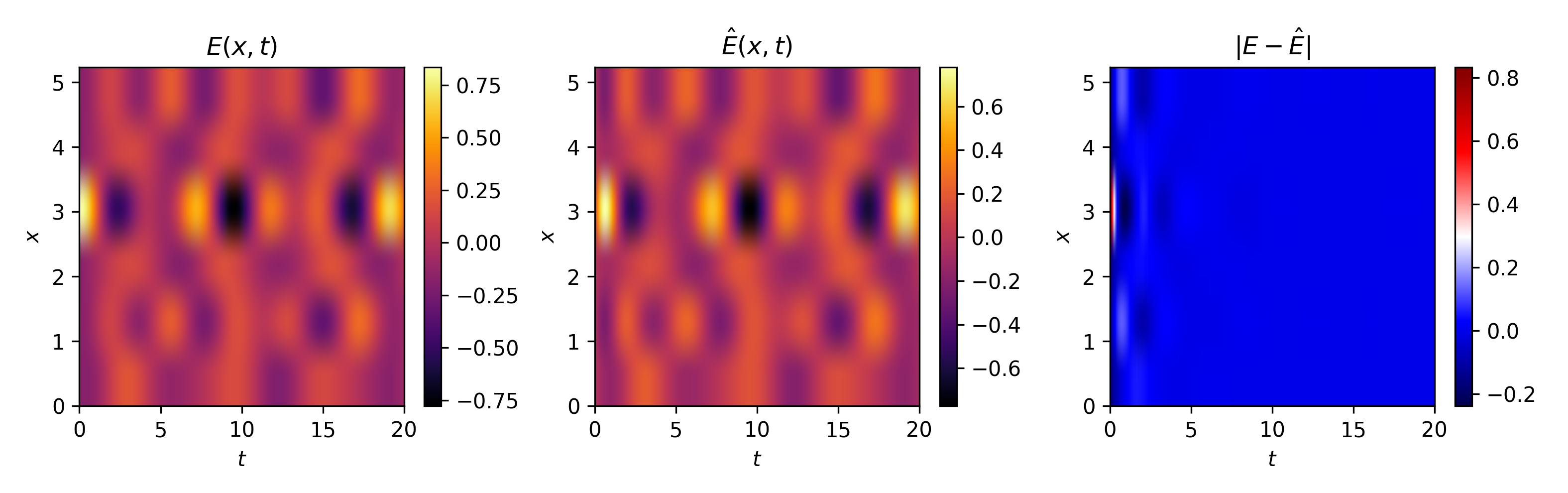}
  \caption{Space--time evolution of $E$, $\hat E$, and $|E-\hat E|$.}
  \label{fig:E_spacetime}
\end{figure}

\begin{figure}[!ht]
  \centering
  \includegraphics[width=\linewidth]{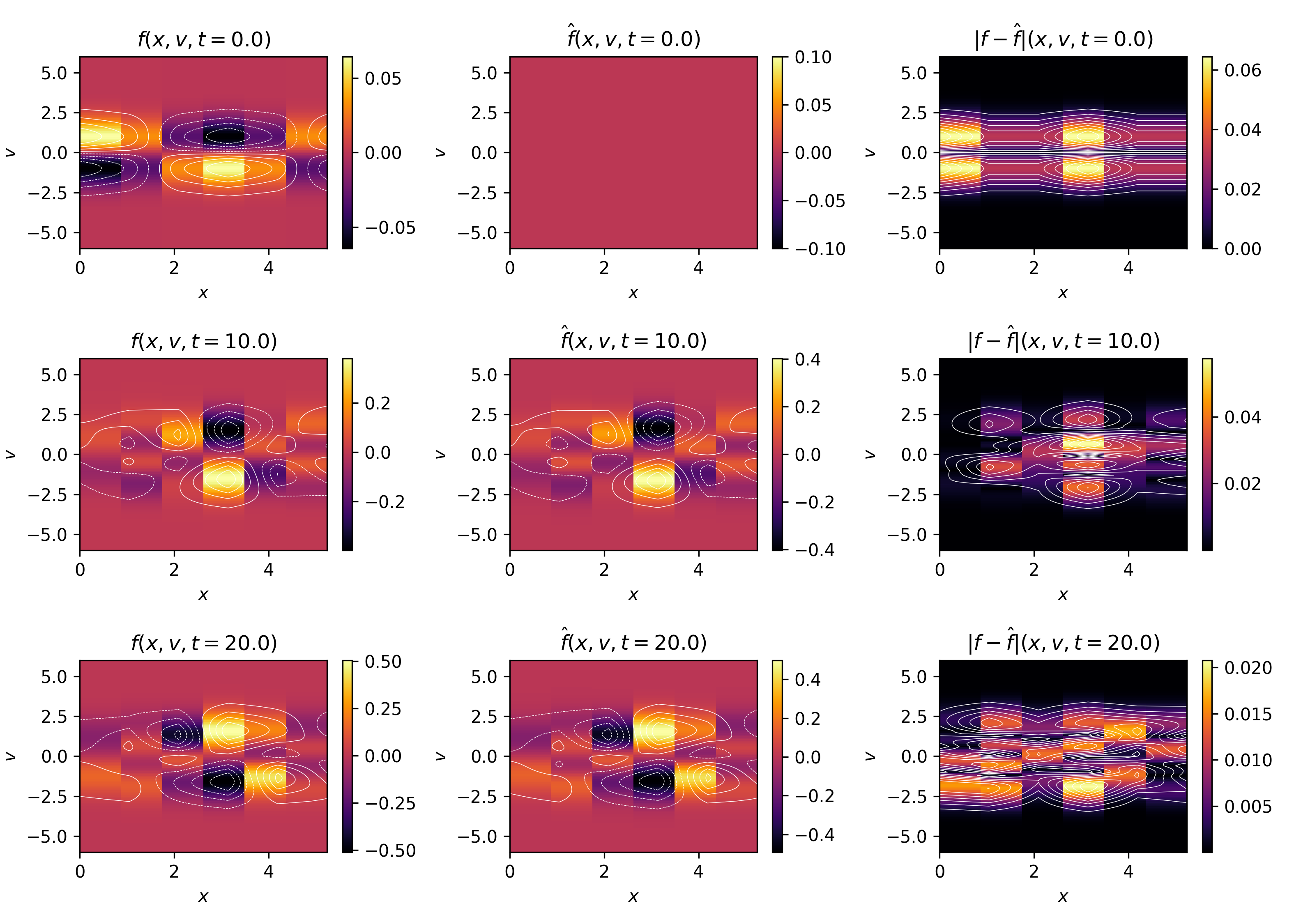}
  \caption{Phase-space snapshots of $f$ and $\hat f$ ($\mathcal{H}_\infty$).}
  \label{fig:phase_space}
\end{figure}

\begin{figure}[!ht]
  \centering
  \includegraphics[width=.7\linewidth]{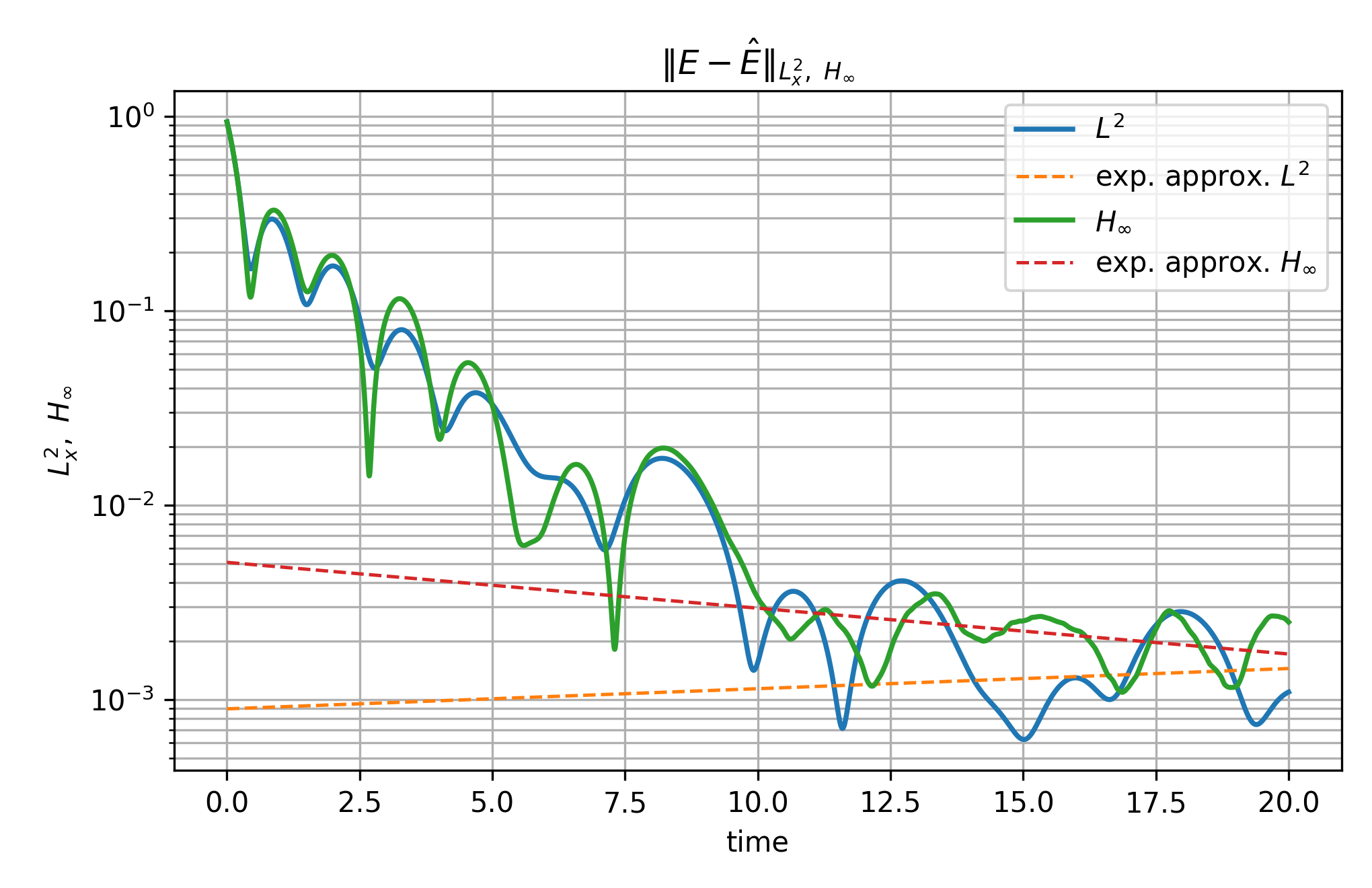}
  \caption{Time evolution of $\|E-\hat E\|_{L^2_x}$ (semilogarithmic scale) for
  the $L^2$ and $\mathcal{H}_\infty$ designs.}
  \label{fig:err_curves_E}
\end{figure}

\begin{figure}[!ht]
  \centering
  \includegraphics[width=.7\linewidth]{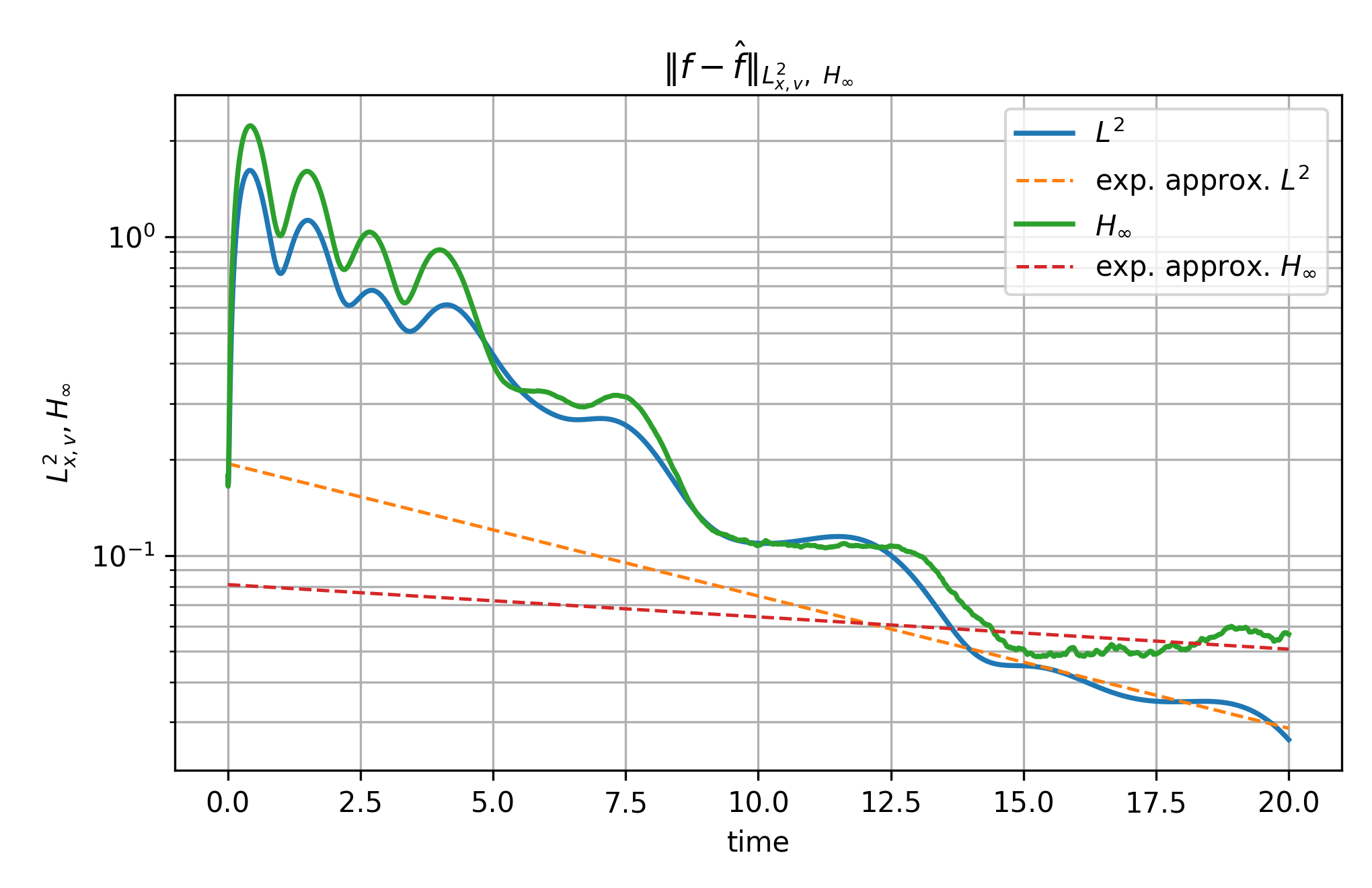}
  \caption{Time evolution of $\|f-\hat f\|_{L^2_{x,v}}$ computed from Hermite
  coefficients (semilogarithmic scale) for the $L^2$ and $\mathcal{H}_\infty$
  designs.}
  \label{fig:err_curves_f}
\end{figure}

\subsection{Error analysis}

The estimation errors are measured in the norms
\(
\|E-\hat E\|_{L^2_x},\;\,  \|f-\hat f\|_{L^2_{x,v}},
\)
where the kinetic norm is evaluated consistently with the Hermite expansion,
\[
\|f-\hat f\|_{L^2_{x,v}}^2
\approx
\Delta x\sum_x\sum_{n=0}^{N_v-1}|a_n(x)-\hat a_n(x)|^2.
\]

This representation enables a direct comparison of the spectral coefficients $a_n(x)$ and their estimates $\hat a_n(x)$. Both norms are computed at each time step to quantify convergence and transient behavior.

Figure~\ref{fig:err_curves_E} reports the time evolution of
$\|E-\hat E\|_{L^2_x}$ for both designs, while Figure~\ref{fig:err_curves_f} displays the corresponding kinetic error $\|f-\hat f\|_{L^2_{x,v}}$.

At $t=20$, the observed magnitudes are
\[
\|E-\hat E\|_{L^2_x}\approx10^{-3},\qquad
\|f-\hat f\|_{L^2_{x,v}}\approx3\times10^{-2}.
\]

Exponential rates estimated from the asymptotic regime are
\[
\begin{aligned}
&\text{$L^2$ design:} && \lambda_E=-2.39\times10^{-2},\quad \lambda_f=-9.56\times10^{-2},\\
&\mathcal{H}_\infty\text{ design:} && \lambda_E=-5.43\times10^{-2},\quad \lambda_f=-2.33\times10^{-2},
\end{aligned}
\]
confirming exponential convergence with distinct robustness--performance trade-offs.

\begin{remark}
The magnetic field $B$ is not displayed separately. In the present reduced setting, $E$ and $B$ are dynamically coupled through Maxwell's equations, so that a separate reconstruction of $B$ would provide largely redundant
information and would not bring additional insight for the validation of the observer performance.
\end{remark}

\section{Conclusion}
\label{sec:conclusion}

A parameter-dependent operator formulation was developed for control and observer analysis of the Vlasov--Maxwell system.
The linearized kinetic--electromagnetic dynamics were modeled as a non-autonomous evolution system with well-posedness and uniform growth properties.
Parameter-dependent Lyapunov operators yielded operator differential LMIs ensuring exponential stability, observer convergence, and $H_\infty$ disturbance attenuation.
Galerkin projections provided finite-dimensional LMIs consistent with the underlying operator inequalities.
Numerical results on a reduced Vlasov--Maxwell benchmark illustrated the effectiveness of the proposed synthesis methodology and confirmed the predicted convergence properties.

\appendix

\section{Auxiliary Results Used in the Proofs}
\label{app:auxiliary}

This appendix recalls the functional-analytic results explicitly used in the well-posedness and Galerkin-consistency arguments.

\subsection{Kato's theorem for non-autonomous evolution equations}

\begin{theorem}[Kato~\cite{kato1953integration}]
\label{thm:Kato}
Let $\{A(t)\}_{t\ge0}$ be a family of closed, densely defined operators on a Hilbert space $\mathscr{X}$ with common domain $\mathcal D$, satisfying:
\begin{enumerate}
\item[(i)] each $A(t)$ generates a $\mathcal C_0$-semigroup;
\item[(ii)] there exist $M\ge1$ and $\omega\in\mathbb R$ such that
\(
\|(A(t)+\lambda I)^{-1}\|
\le
M(\lambda-\omega)^{-1},
\quad \lambda>\omega,
\)
uniformly in $t$;
\item[(iii)] $t\mapsto A(t)x$ is continuous for every $x\in\mathcal D$.
\end{enumerate}
Then the non-autonomous Cauchy problem
\(
\dot X(t)=A(t)X(t), \; X(s)=X_s,
\)
 admits a unique evolution family
$\{\mathcal S(t,s)\}_{t\ge s\ge0}$ satisfying
\(
\mathcal S(t,s)\mathcal S(s,r)=\mathcal S(t,r),
\;
\mathcal S(s,s)=I,
\;
\|\mathcal S(t,s)\|\le Me^{\omega(t-s)} .
\)
Moreover, $X(t)=\mathcal S(t,s)X_s$ is the unique mild solution for every
$X_s\in\mathscr X$.
\end{theorem}

\subsection{Stone's theorem for unitary groups}

\begin{theorem}[Stone~\cite{Stone1932,EngelNagel2000}]
\label{thm:stone}
Let $(U(t))_{t\in\mathbb{R}}$ be a strongly continuous one-parameter unitary
group on a Hilbert space $\mathscr{X}$. Then there exists a unique self-adjoint operator $A:\mathcal{D}(A)\subset\mathscr{X}\to\mathscr{X}$ such that
\(
U(t)=e^{\,itA},
\quad t\in\mathbb{R}.
\)
Conversely, if $A$ is self-adjoint, then
\(
U(t):=e^{\,itA},
\quad t\in\mathbb{R},
\)
defines a strongly continuous unitary group on $\mathscr{X}$.
\end{theorem}

\subsection{Banach--Alaoglu compactness theorem}

\begin{theorem}[Banach--Alaoglu~\cite{banach1932theorie,alaoglu1940weak}]
\label{thm:banach-alaoglu}
Let $X$ be a normed vector space. Every bounded subset of the dual space $X^{\star}$ is relatively compact in the weak-$\star$ topology. Equivalently, the closed unit ball
\(
B_{X^{\star}}
=
\{\ell\in X^{\star}:\|\ell\|\le 1\}
\)
is weak-$\star$ compact.
\end{theorem}

\section*{conflict of interest}
No potential conflict of interest was reported by the author(s).

\section*{Data availability statement}
The authors declare that this work is self-contained and there is no additional data is used.
\section*{ORCID}
Amadou Cissé https://orcid.org/0000-0002-8014-4180


\begin{thebibliography}{10}
\providecommand \doibase [0]{http://dx.doi.org/}%

\bibitem{GlassHanKwan2012}
Glass O, Han-Kwan D. On the controllability of the relativistic Vlasov--Maxwell
  system. {\it Journal de Math{\'e}matiques Pures et Appliqu{\'e}es.}
  2015\string;103(3)\string:695--740.

\bibitem{Weber2021}
Weber J. Optimal control of a two-dimensional Vlasov--Maxwell system. {\it
  ESAIM: Control, Optimisation and Calculus of Variations.}
  2021\string;S19\string:35.

\bibitem{cisse2020observers}
Cisse A, Boutayeb M. Observers of vlasov-poisson system. {\it
  IFAC-PapersOnLine.} 2020\string;53(2)\string:5946--5951.

\bibitem{cisse2024software}
Ciss{\'e} A, Boutayeb M. Software sensors design for a class of non linear
  coupled PDE systems: The Vlasov--Poisson dynamical system. {\it Mathematics
  and Computers in Simulation.} 2024\string;219\string:284--296.

\bibitem{cisse2024state}
Ciss{\'e} A, Boutayeb M. On state feedback control of a class of NonLinear PDE
  systems in finite dimension. {\it Communications in Nonlinear Science and
  Numerical Simulation.} 2024\string;130\string:107751.

\bibitem{knopf2020optimal}
Knopf P, Weber J. Optimal control of a Vlasov--Poisson plasma by fixed magnetic
  field coils. {\it Applied Mathematics \& Optimization.}
  2020\string;81(3)\string:961--988.

\bibitem{bartsch2024controlling}
Bartsch J, Knopf P, Scheurer S, Weber J. Controlling a Vlasov--Poisson Plasma
  by a Particle-in-Cell Method Based on a Monte Carlo Framework. {\it SIAM
  Journal on Control and Optimization.} 2024\string;62(4)\string:1977--2011.

\bibitem{chesi2005polynomially}
Chesi G, Garulli A, Tesi A, Vicino A. Polynomially parameter-dependent Lyapunov
  functions for robust stability of polytopic systems: an LMI approach. {\it
  IEEE transactions on Automatic Control.} 2005\string;50(3)\string:365--370.

\bibitem{a2019computational}
A.~Mozelli L, S.~Adriano RL. On computational issues for stability analysis of
  LPV systems using parameter-dependent Lyapunov functions and LMIs. {\it
  International Journal of Robust and Nonlinear Control.}
  2019\string;29(10)\string:3267--3277.

\bibitem{jetto2010efficient}
Jetto L, Orsini V. Efficient LMI-based quadratic stabilization of interval LPV
  systems with noisy parameter measures. {\it IEEE Transactions on Automatic
  Control.} 2010\string;55(4)\string:993--998.

\bibitem{wu2001lpv}
Wu F, Grigoriadis KM. LPV systems with parameter-varying time delays: analysis
  and control. {\it Automatica.} 2001\string;37(2)\string:221--229.

\bibitem{apkarian2000parameterized}
Apkarian P, Tuan HD. Parameterized LMIs in control theory. {\it SIAM journal on
  control and optimization.} 2000\string;38(4)\string:1241--1264.

\bibitem{coron2007strict}
Coron JM, Novel dB, Bastin G. A strict Lyapunov function for boundary control
  of hyperbolic systems of conservation laws. {\it IEEE Transactions on
  Automatic control.} 2007\string;52(1)\string:2--11.

\bibitem{mironchenko2017characterizations}
Mironchenko A, Wirth F. Characterizations of input-to-state stability for
  infinite-dimensional systems. {\it IEEE Transactions on Automatic Control.}
  2017\string;63(6)\string:1692--1707.

\bibitem{pezzi2018velocity}
Pezzi O, Servidio S, Perrone D, et al. Velocity-space cascade in magnetized
  plasmas: Numerical simulations. {\it Physics of Plasmas.} 2018\string;25(6).

\bibitem{kamaletdinov2024nonlinear}
Kamaletdinov SR, Vasko IY, Artemyev AV. Nonlinear electron scattering by
  electrostatic waves in collisionless shocks. {\it Journal of Plasma Physics.}
  2024\string;90(2)\string:905900201.

\bibitem{ghizzo2024collisionless}
Ghizzo A, Del~Sarto D, Betar H. Collisionless heating in Vlasov plasma and
  turbulence-driven filamentation aspects. {\it Physics of Plasmas.}
  2024\string;31(7).

\bibitem{brizard2007foundations}
Brizard AJ, Hahm TS. Foundations of Nonlinear Gyrokinetic Theory. {\it Reviews
  of Modern Physics.} 2007\string;79(2)\string:421--468.

\bibitem{arro2022spectral}
Arr{\`o} G, Califano F, Lapenta G. Spectral properties and energy transfer at
  kinetic scales in collisionless plasma turbulence. {\it Astronomy \&
  Astrophysics.} 2022\string;668\string:A33.

\bibitem{svidzinski2024full}
Svidzinski V, Zhao L, Kim J, Barov N. Full wave modeling of radio-frequency
  beams in tokamaks in the electron cyclotron frequency range. {\it Physics of
  Plasmas.} 2024\string;31(4).

\bibitem{zaar2025enhanced}
Zaar B, Johnson T, B{\"a}hner L, Bilato R, Ragona R, Vallejos P. Enhanced ion
  heating using a TWA antenna in DEMO-like plasmas. {\it Journal of Plasma
  Physics.} 2025\string;91(1)\string:E13.

\bibitem{kato1953integration}
Kato T. Integration of the equation of evolution in a Banach space. {\it J.
  Math. Soc. Japan.} 1953\string;5\string:208--234.

\bibitem{Stone1932}
Stone MH. On one-parameter unitary groups in {H}ilbert space. {\it Annals of
  Mathematics.} 1932\string;33(3)\string:643--648.

\bibitem{EngelNagel2000}
Engel KJ, Nagel R. {\it One-Parameter Semigroups for Linear Evolution
  Equations}. 194 of {\it Graduate Texts in Mathematics}.
\newblock Springer, 2000.

\bibitem{banach1932theorie}
Banach S. {\it Th{\'e}orie des Op{\'e}rations Lin{\'e}aires}.
\newblock Monografje Matematyczne, 1932.

\bibitem{alaoglu1940weak}
Alaoglu L. Weak topologies of normed linear spaces. {\it Annals of
  Mathematics.} 1940\string;41(1)\string:252--267.

\end{thebibliography}
\end{document}